\documentclass[12pt,onecolumn]{IEEEtran}
\usepackage{graphicx,amsmath,amssymb}
\usepackage{url}

\ifCLASSOPTIONcompsoc
   \usepackage[tight,normalsize,sf,SF]{subfigure}
\else
   \usepackage[tight,footnotesize]{subfigure}
\fi
\newtheorem{definition}{Definition}
\newtheorem{theorem}{Theorem}
\newtheorem{proposition}{Proposition}

\newtheorem{lemma}{Lemma}

\usepackage{color,ifthen,graphicx,enumerate,float,url}
\usepackage{mathrsfs}
\usepackage{nicefrac}
\usepackage{pstool}
\usepackage{bm}
\DeclareMathOperator*{\argmax}{arg\,max}
\DeclareMathOperator*{\argmin}{arg\,min}

\renewcommand{\t}{\theta}

\renewcommand{\t}{\theta}

\newcommand{\dif}{\mathrm{d}}
\begin{document}
\title{Dynamic routing for social information sharing}
\author{Yunpeng~Li,~Costas~Courcoubetis,~and~Lingjie~Duan
\thanks{Y. Li, C. Courcoubetis and L. Duan are with the Engineering Systems and Design Pillar, Singapore University of Technology and Design, Singapore 487372,
Singapore (e-mail: yunpeng\_li@mymail.sutd.edu.sg, \{costas, lingjie\_duan\}@sutd.edu.sg).
}}
\maketitle

\begin{abstract}
 Today mobile users are intensively interconnected thanks to the emerging mobile social networks, where they
 share location-based information with each other when traveling on different routes and visit different areas of the city.
 In our model the information collected is aggregated over all users' trips and made publicly available as a public good.
 Due to information overlap, the total useful content amount increases with the diversity in path choices made by the users,
 and it is crucial to motivate selfish users to choose different paths despite the potentially higher costs associated with their trips.
 In this paper we combine the benefits from social information  sharing with the fundamental routing problem
 where a unit mass of non-atomic selfish users decide their trips in a non-cooperative game by choosing between a high-cost and a low-cost path.
 To remedy the inefficient low-content equilibrium where all users choose to explore a single path (the low-cost path),
 we propose and analyse two new incentive mechanisms that can be used by the social network application,
 one based on side payments and the other on restricting access to content for users that choose the low cost path.
 Under asymmetric information about user types (their valuations for content quality),
 both mechanisms efficiently penalise the participants that use the low-cost path and reward the participants that take the high-cost path.
 They lead to greater path diversity and hence to more total available content at the social cost of reduced user participation
 or restricted content to part of the users.
 We show that user heterogeneity can have opposite effects on social efficiency depending on the mechanism used.
 We also obtain interesting price of anarchy results that show some fundamental tradeoffs between achieving path diversity and
 maintaining greater user participation, motivating a combined mechanism to further increase the social welfare.
Our model extends classical dynamic routing in the case of externalities caused from traffic on different paths of the network.
\end{abstract}

\section{Introduction}
Emerging mobile social networks are developing fast to strengthen mobile users' social ties and allow them to share useful location-based information with others. The information is usually aggregated over all users' trips and made publicly available to benefit all users that participate.
The shared information (e.g., about shopping promotions and locations, restaurant discovery, air quality, traffic conditions, etc.)
weigh heavily on users' daily life, and is considered a public good that benefits all.
Plenty of such information sharing applications (e.g., TripAdvisor and Yelp\footnote{See www.tripadvisor.com and www.yelp.com})
have been proposed to involve millions of mobile users,
and the revenue of the related location-based services is expected to increase to around US\$40 billions by 2019 \cite{lbs}.
For example, Microsoft has recruited workers through Gigwalk, a famous mobile platform, to photograph
3D panoramas of businesses and restaurants, which are integrated to Bing Maps data \cite{GigwalkRicardo}.
Another example is that to collect air quality data, mobile participatory sensing systems are developed
that request participants to choose diverse routes to sense \cite{ganti2011mobile}.

Social information sharing via mobile routing and sensing comes at a cost to the mobile users.
Travelling on a path generally incurs a travel cost, e.g., cost of time and gas  to a user.
When making routing decisions to sense information, users are selfish and may only choose the routes with such minimum costs.
Hence, incentive mechanisms are crucial for inducing larger participation and motivating individuals to perform diverse routing and sensing tasks that are more valuable
to the community. Such mechanisms, designed by a ``social planner", must consider the various costs and benefits related to the potential information collected by the participants. Prior routing game literature did not study any social benefit from information sharing among users,
but only looked at users' travel costs (e.g., congestion generated by users on the same route)
for equilibrium routing analysis.
The prior literature mainly used concepts Nash equilibrium for best-response users or Wardrop equilibrium for non-atomic users (e.g., \cite{roughgarden2002bad}, \cite{richman2007topological}, \cite{acemoglu2006paradoxes}, \cite{altman2001routing}), and analyzed price of anarchy under different pricing schemes to control congestion.
Differently, this paper studies how to regulate dynamic routing that includes not only the travel costs along the paths but also
the benefit users get from the shared information collected along all paths.
Hence, our model extends in a fundamental way the traditional dynamic routing setup to include
the \emph{positive externalities} generated by users traveling on different paths:
a user traveling on one path benefits from the content collected by users traveling on another path.

We also note that this paper is different from the recent mobile crowdsensing studies (e.g., \cite{ganti2011mobile}, \cite{yang2012crowdsourcing}, \cite{chorppath2013trading}, \cite{duan2014motivating}, \cite{cheung2015distributed}, ), which focused on the analysis of participation incentives via auction, contract or pricing in a principal-agent structure under asymmetric information. In our model users are non-atomic,
and while travelling they collect information that is useful to all, acting as both information contributors and consumers.
Borrowing concepts from standard public good models (e.g., \cite{courcoubetis2006incentives}, \cite{courcoubetis2012economic}), our model combines information sharing with routing decision incentives to collect information from different routes when this
has a high informational value to the community.
Our aim is to design optimal incentive mechanisms for heterogeneous users
to use paths with potentially higher costs in order to increase the level of public good provisioning, i.e., the diversity and value of information available to the community.

Our basic content sharing model starts with a single path. Participants using this path collect information
that is aggregated by the system planner in a sub-additive way (due to possible information overlaps),
and made available to
the community as a public good with controlled access.
The valuation of this public good may depend on a user's type that is private information.
Under asymmetric information, the social planner knows the fraction of users in each type but cannot distinguish between individual users.
In the case of multiple paths, information collected from different paths is independent, and hence
the same number of users traveling over different paths
generates more total content compared to traveling over a single path.

As shown in Fig.~\ref{model}, our network routing model is the simplest possible involving only two paths (similar to the fundamental routing problem in \cite{altman2001routing}):
$H$ with some non-zero high cost, and $L$ with low cost, for simplicity equal to zero.
Though simple, this network model captures users' heterogeneity in travel cost and content valuation,
and allows us to explain intuitively and clearly the loss of efficiency due to users' private information that various incentive policies exhibit.
Adding a non-linear congestion cost term in each path does not change the qualitative properties of our equlibrium results,
but makes the analysis unnecessarily more complex.

\begin{figure}[!t]
  \centering
  \includegraphics[width=2.5in]{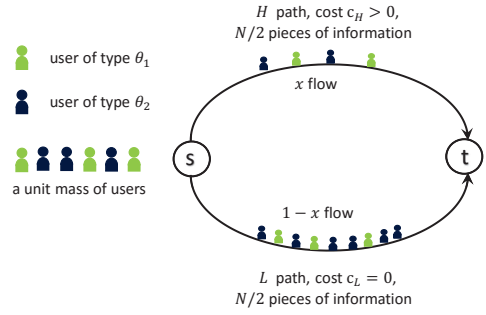}
  \caption{A unit mass of users travels from source $s$ to destination $t$. Travelling on the (high cost) $H$-path generates to a user a cost of $c_H>0$, while travelling on the (low cost) $L$-path generates a cost of $c_L=0$. Each path contains $N/2$ pieces of information.}
  \label{model}
\end{figure}

We analyse a non-cooperative game where a unit mass of selfish non-atomic users must decide
their trips by choosing between the two paths.
Our model extends traditional routing game models (e.g., \cite{roughgarden2002bad}, \cite{richman2007topological}, \cite{acemoglu2006paradoxes}) by considering both the costs of the trips and the information value obtained by having access to the total amount of collected content.
Each user has a ``type" regarding her valuation for content:
how much she benefits from the content collected (per unit),
or equivalently, in the case of using an application processing this content,
how the quality of the calculations depend on the amount of content made available.
Now, each user obtains a net utility that is proportional to her type (defined above)
multiplied by the total amount of content collected,
minus her travel cost.
Without any incentive, since each non-atomic user's content contribution is infinitesimal compared to the total  amount of content collected,
the only equilibrium is for all users to take the low-cost path. However, the social optimum in general can
involve a fraction of users taking the high-cost path to increase significanlty path diversity and hence total content available, even if this is more costly.

For this system we address the problem of constructing incentive  mechanisms for motivating
the optimal fraction of users to take the high-cost path,  achieving
the best possible tradeoff between travel cost and content value.
Our benchmark for evaluating social efficiency is the centralized model where the social planner dictates the optimal routing actions for all the users.
We analyse the properties of two individually rational, incentive compatible and budget balanced incentive mechanisms,
one using side payments and the other restricting access to content.
Side payments are collected from users choosing the low-cost path and create
subsidies to users that take the high-cost path, to keep the budget balanced.
We expect such a mechanism to be implemented by the content sharing social network application;
for example, by reducing membership fees to users that take the high-cost path.

Content-restriction is implemented by controlling
the fraction $a<1$ of the content made accessible to users that take the low-cost path
(same as ``destroying" a fraction $1-a$ of the content being collected),
and does not require a billing system as side-payment does.
Both schemes penalise directly or indirectly the users that choose the low-cost path and reward the users that take the high-cost path.
We capture the important trade-off between achieving greater path diversity at the expense of lower user participation or higher content destruction,
and show that higher user diversity in their valuation for content can have opposite effects in the two mechanisms.

Though simple, this network model captures heterogeneity (i) of paths in travel cost and (ii) of users in content valuation.
It shows the loss of efficiency due to users' private information that reduces the capabilities of the social planner (application) to control routing decisions.

Our main findings are summarised as follows.
\begin{itemize}
\item
Side payments inflict the same cost/benefit to all user independently of their types
and there is a tension between raising the cost for users choosing the low-cost path
and violating the participation condition of low valuation users.
It turns out that when users have homogeneous valuations for content, side payments
don't suffer from this participation issue and achieve  social optimality (with price of anarchy ($PoA$) equal to 1),
while when users have diverse content valuations, under the optimum incentive
scheme users of low-valuation type may not participate, reducing efficiency and $PoA$ to 1/2 from the social optimum
where the social planner has full power.
\item
By restricting content, users who value content more have a stronger incentive to take the high-cost path.
This incentive scheme affects more the high-valuation user type, even if this is private information.
Contrary to the side-payment case mentioned earlier, user diversity improves performance since high valuation type users may be willing to move to the H path even if information restriction is low on the L path.
Hence under some conditions on user type distribution, content restriction can be more efficient than
side payments. Again, it achieves a $PoA$ 1/2.
\item
Combining the two incentive schemes produces strictly better results than using any single one.
Applying a degree of content restriction reduces the amount of side payments needed to induce path diversity, hence increasing low-valuation users' participation.
We show that the $PoA$ is beyond 0.7. 
\item
Robustness of the results concerning the two incentive schemes is shown by extending our model
to a more general network model including more than two paths  and with possible path overlap.
We successfully address the new challenges of these models including multiple and unstable equilibria
in the optimal incentive design.
\end{itemize}

The rest of the paper is organized as follows. Section~\ref{systemmodel} introduces the system model when crowdsensing meets routing. Sections \ref{sidepayment} and \ref{contentrestriction} present and analyze the two kinds of mechanisms. Section \ref{combined} shows the combined mechanism. Section \ref{simulation} shows simulation results. Section \ref{generalization} and Section \ref{extensions} show some extensions of our model. Section \ref{conclusion} concludes the paper.  

\section{System Model and Problem Formulation}\label{systemmodel}
When a user travels, she continuously senses information that is useful for the user community.
This information is aggregated over all users and offered back to the user community as
a public good.
We consider the network model from source node $s$ to destination node $t$ in Fig.~\ref{model},
where there are two paths 
from the set $\mathcal{P}:=\{H,L\}$. Though simple, the two-path network model still captures all the interesting aspects of incentive-based routing policies in a clear and educational way. A more general network with multiple paths and path overlap is further analyzed in Section \ref{generalization}.

Travelling on the $H$-path ($L$-path) incurs a travel cost, e.g., cost of time and gas, $c_H$ ($c_L$) to a user, where $c_H>c_L$.
For simplicity,  we normalize these costs so that $c_H>0$ and $c_L =0$. Such costs are determined by well-known exogenous factors like traveling distance, average congestion during time of travel, and quality of road surface, which are therefore common for all users. 
Note that the travel costs per user are fixed and
do not depend on the amount of traffic along the paths. This is sensible our system users (e.g., in Gigwalk) are assumed to represent a small fraction of the total population that travels.\footnote{When future users no longer represent a small fraction, we will consider traffic-related costs depending on users' routing decisions as in the routing game literature (e.g., \cite{roughgarden2002bad}, \cite{richman2007topological}, \cite{acemoglu2006paradoxes}).
}

Given the two path candidates $H$-path and $L$-path, we except a large number $n$ of users
that make selfish routing decisions in a one-shot game. 
Since $n$ is assumed large for routing, we reasonably consider an individual user to be non-atomic as in most of the routing game literature
(e.g., \cite{roughgarden2002bad}, \cite{richman2007topological}, \cite{acemoglu2006paradoxes}), and model our $n$ users by a single infinitely divisible
unit mass of non-atomic users.
In this one-shot game, each user chooses a path $P\in\mathcal{P}$.
We denote by $x_H$ the fraction of the users that choose the $H$-path, leaving $1-x_H$ to use the $L$-path.

Users obtain a net benefit that is determined by travel costs and the information collected.
To model this information value we need to define an information or content model, i.e., how different user routing decisions
affect the total information collected. We present now a specific model
for information collection that motivates the general properties
that such information collection models should possess.

 For simplicity, we model the total information available in the network as a set of $N$ different information pieces uniformly
distributed along the two paths, i.e., each path has $N/2$  independent information pieces with no information overlap
between paths.\footnote{Our model and analysis can be easily extended to any information partition between the two paths.}
Since paths are symmetric in term of available information, we denote by\footnote{We use the subscript 1 in $Q$ to denote the content collected from a single link.} $Q_1(x)$ the expected number of independent
pieces of information collected over an arbitrary path if (i)  a fraction $x$ of the $n$ users (or a mass
$0\leq x\leq 1$ among the total unit mass of users) travel over that path,  (ii)
each user collects $\phi$ distinct information pieces among the $N/2$ available,
and (iii) each user selects her information items independently of other users.
In this model each information piece has probability $\phi/(N/2)$ to be sensed by a given user,
and probability $1-(1-\phi/(N/2))^{nx}$ to be sensed by at least one among the $n x$ users.
Hence, the total average number of items sensed by our $nx$ users (the average ``information content'')
on this single link is
\begin{equation}
\label{Q}
Q_1(x)=\frac{N}{2}\big(1-(1-\frac{2\phi}{N})^{nx}\big) \,,
\end{equation}
which is concave and increasing in $x$.
Similarly, the information value on the other path is $Q_1(1-x)$ since $n(1-x)$ users will take this second path.
We denote by $Q(x,1)$ the total information content collected and aggregated from both paths when a fraction $x$ of the
total unit mass of users choose
one of the paths, which is given by
\begin{equation}
\label{content1}
Q(x,1)=Q_1(x)+Q_1(1-x)\,.
\end{equation}
This  is concavely increasing in $x$ when $0\le x\le 0.5$ and decreasing when $0.5<x\leq 1$.
As not all users may participate, we may extend this notation and denote by
\begin{equation}
\label{eq:b}
Q(x,b)=Q_1(x)+Q_1(b-x),
\end{equation}
where the total information content
obtained if a mass $x$ out of a total of $b$ chooses one of the paths (hence $b-x$ chooses the other path). Hence, $b$ is the active participation fraction among users and the shape of our content function depends on the richness of the content (value of $N$) and the actual member of users that are abstracted into a continuous mass.

The above example motivates the properties of more general information content functions $Q(x,1)$
that are based on different forms of the $Q_1(x)$ function. We require that $Q(x,1)$:
(i) increasing and concave for $0\le x\le 0.5$, and (ii) $Q(x,1) = Q(1-x,1)$.
This implies that creating a more balanced traffic in the two links  increases the total sensed information.
Our results in the later sections do not depend on the specific model in (\ref{Q}), assuming any general (but ``sensible'') model
of information content mentioned above.
In the rest of the paper we use the $H$-path as the reference path, corresponding to the use of  $x = x_H$ in the above formulas.
Hence, $Q(x_H, 1)$ is the total content available assuming that $x_H$ traffic goes through the $H$-path.

We next address that users may have different valuations for content, i.e., have different sensitivities to the
total amount of content $Q(x_H, 1)$ made available. For example, a user
might have little use of the total content or require a small fraction of this content (equivalently, care less about the quality of the service provided by the application).
We model this by a multiplicative parameter $\theta$.
To make our analysis simpler,
we consider two valuation types among users, 
denoted by $\theta_1$ and $\theta_2$,
where $\theta_1\leq \theta_2$.
An important parameter of the model is the average valuation $\theta_0$ given by
\begin{equation}
\label{type}
\theta_0=\eta\theta_1+(1-\eta)\theta_2\,,
\end{equation}
where $\eta$ and $1-\eta$ are the user proportions in low and high valuations respectively.
We assume that user types are private information (i.e., only a user knows its type $\t_i$) and that
the system only knows the value of $\eta$ (i.e., the fraction of each type in the total population).
If a user with type $i\in\{1,2\}$ chooses path $P\in\{H,L\}$, her payoff is the difference between her perceived information value $\theta_iQ(x_H,1)$ and the travel cost on path $P$. That is,
\begin{equation}
\label{payoff}
u_i(P,x_H)=\left\{
\begin{array}{ll}
\theta_iQ(x_H,1), &\mbox{if }P=L,\\
\theta_iQ(x_H,1)-c_H, &\mbox{if }P=H,
\end{array}
\right.
\end{equation}
and the user will prefer the path with a higher payoff. Observe that in our non-atomic user model
 an individual user has an infinitesimal contribution to $Q(x_H,1)$,
 thus her perceived information value $\theta_iQ(x_H,1)$ is not affected by her path choice $P$.
 But it depends on the choices made by the rest of the users through the value of $x_H$.
We naturally formulate such interaction among users in path choosing as a \emph{content routing game}.
We denote a feasible flow's path partition  by $(x_H,1-x_H)$, or just by $x_H$ since the flow on the $L$-path
follows uniquely from $x_H$. The next definition is the equivalent of a Wardrop equilibrium (we place a hat over a symbol to denote equilibrium flow).
\begin{definition}
\label{def1}
A feasible flow $\hat{x}_H$ in the content routing game is an equilibrium if no user traveling over the $H$-path or the $L$-path will
profit by  deviating from her current path choice to increase her payoff.
\end{definition}

The social welfare (total system efficiency) is defined as
\begin{eqnarray}
SW(x_H)&=&\eta\theta_1Q(x_H,1)+(1-\eta)\theta_2Q(x_H,1)-x_Hc_H\nonumber\\
&=&\theta_0Q(x_H,1)-x_Hc_H.\label{eq:SW}
\end{eqnarray}
If we can perfectly control all users' decisions in a centralized way, we can achieve the social optimum flow $x_H^*$
that maximises $SW(x_H)$ in \eqref{eq:SW} without any constraint:
\begin{equation}
\label{swl}
x_H^*\in\argmax\limits_{0\le x_H\le1}\{\theta_0Q(x_H,1)-x_Hc_H\}\,.
\end{equation}
(We use superscript $^*$  to denote optimal values).
It is easy to see that $x_H^*\in[0,0.5]$,
otherwise we can replace it by $1-x_H^*$ to achieve the same information value in (\ref{content1}) at a smaller total cost on the $H$-path.
\subsection{Equilibrium without Incentive Design}
In real world users are selfish and may not behave optimally without appropriate incentives.
We now analyse the users' behavior in the content routing game without any added incentives.
Notice  no matter which path a user of type-$i$ chooses, the information value she perceived is always $\theta_iQ(x_H,1)$ according to (\ref{payoff}) and hence her choice only depends on path cost.
Therefore, the selfish routing strategy is for every user to choose the $L$-path at the equilibrium.
\begin{proposition}\label{prop}
There exists a unique equilibrium for our content routing game: $\hat{x}_H=0$ and the resultant social welfare is $SW(0)=\theta_0Q(0,1)$.
\end{proposition}

There is a gap between optimal social welfare and the social welfare attained at
the above equilibrium. We measure the gap by \emph{price of anarchy} ($PoA$) \cite{nisan2007algorithmic},
which is the ratio between the lowest social welfare at any equilibrium and the optimal social welfare $SW(x_H^*)$,
 by searching over all possible system parameters and $Q(\cdot, 1)$ functions. In our content routing game, the specific formula of $PoA$ without any added incentives is
$$PoA=\min\limits_{\theta_i,Q(\cdot,1),c_H}\frac{\theta_0Q(0,1)}{SW(x_H^*)}.$$
\begin{proposition}
\label{prop2}
The price of anarchy of the content routing game without incentive design is $PoA=1/2.$
\end{proposition}
\begin{proof}
Note that the maximal quantity of content is attained by letting half flow go through the high cost path and half go through the low cost path, hence
$$SW(x_H^*)\le\max\limits_{0\le x_H\le1}\{\theta_0Q(x_H,1)\}=\theta_0Q(0.5,1).$$
Thus the price of anarchy is lower-bounded as follows:
$$PoA\ge\min\limits_{\theta_i,Q(\cdot,1),c_H}\frac{\theta_0Q(0,1)}{\theta_0Q(0.5,1)}=\min\limits_{Q_1(\cdot)}\frac{Q_1(1)}{2Q_1(0.5)}\ge\frac{1}{2},$$
The last inequality is because that $Q_1(x_H)$ is a nondecreasing function of $x_H$.

We can show that the bound is also tight and hence $PoA=1/2$  by constructing a suitable function $Q_1(x)$. Let
\begin{equation}
\label{fucdef}
Q_1(x)=\left\{
\begin{array}{ll}
2qx &0\le x\le0.5;\\
q &0.5\le x\le1.
\end{array}
\right.
\end{equation}
Then $Q(x,1)=q+2qx$, $0\le x\le 0.5$. This corresponds to having a finite amount of content, and that each user discovers some new piece of content with no overlaps. Consider any fixed value of $c_H$ that is less than $2\theta_0q$, then $x_H^*=0.5$. Using this setup, the $PoA$ for this specific instance becomes
$$PoA=\min\limits_{\theta_i,c_H}\frac{\theta_0Q(0,1)}{SW(0.5)}=\min\limits_{\theta_i,c_H}\frac{\theta_0q}{2\theta_0q-0.5c_H}=\frac{1}{2}.$$
This completes the proof that $PoA=\frac{1}{2}$.
\end{proof}

The $H$-path cost $c_H>0$ motivates all users to take the $L$-path and leads to the inefficient equilibrium.
Hence a sufficiently small but positive cost leads to the biggest loss of efficiency, as the equilibrium ignores
completely the content on the $H$-path.
Realizing the inefficiency, we want to design incentive mechanisms  to approach social optimum $SW(\hat{x}_H)$
with the following properties.
\begin{enumerate}
  \item Individual rationality (IR): The payoff of each $H$- or $L$-path participant should be non-negative to guarantee participation.
  \item Incentive compatibility (IC): Each user with type-$i$ should truthfully decide routing according to her real valuation
  $\theta_i$.
  \item Budget balance (BB): The money collected from some users (if any) should be distributed to the rest of the users. Many crowdsensing applications (e.g., Waze) want to maximize the social welfare and are not profitable in forming a community among users.
\end{enumerate}
\begin{table}[!t]
\centering
\caption{Summary of Key Notations}
\begin{tabular}[tb]{|c|c|}\hline
Symbol & Physical meaning\\\hline
$c_P$&Travel cost on $P$- path, $P\in\{H,L\}$ \\\hline
$N$&Total number of information contained on the two paths \\\hline
$n$&Total number of non-atomic users \\\hline
$x_H$&Fraction of $H$-path participants in $n$ users  \\\hline
$Q_1(x_H)$& $H$-path's content collected by $x_H$ proportion of users \\\hline
$Q(x_H,1)$& Content of both paths given the partition $(x_H,1-x_H)$ \\\hline
$\theta_i$& Type-$i$ users' private valuation towards information, $i\in\{1,2\}$ \\\hline
$\eta$& Fraction of type-$1$ users out of $n$ users \\\hline
$u_i(P,x_H)$& A type-$i$ user's payoff, $i\in\{1,2\}$, by choosing path $P$ under $x_H$ \\\hline
   $SW(x_H)$& Social welfare with $x_H$ proportion of $H$-path participants  \\\hline
    $x_H^*$& Free social welfare maximizer in \eqref{swl}\\\hline
     $\hat{x}_H$& Equilibrium of the content routing game \\\hline
    $b$& Active participation fraction (less than 1) of the unit mass of users\\\hline
    $g(x_H)$& Side-payment function under $x_H$ \\\hline
    $a$& Content-restriction coefficient (less than 1) \\\hline
\end{tabular}
\label{table1}
\end{table}

We will introduce in the following sections two incentive schemes which satisfy the above properties: side-payment  and content-restriction.

We list all the key notations used in this paper in Table \ref{table1}.

\section{Side Payments as Incentive}\label{sidepayment}
The idea of using side-payment as an incentive is to incentivize more users to choose $H$-path by charging more users
that take the $L$-path (and providing positive subsidies to users of the $H$-path).
When applying side-payment (on the $L$-path), it is possible that some users with low-valuation type $\theta_1$
(``low valuation users'') may choose not to participate,
since their payoffs can be negative. In this case, only a proportion $b\leq 1$ of the unit mass of users will participate.
To address this reduction of the total mass of users in the system
we use the definition in \eqref{eq:b}.
We can now formally introduce the side-payment mechanism as follows.

\emph{Side-payment Mechanism:} We charge an extra cost $g(x_H)$ on each of $(b-x_H)n$ $L$-path participants and refund an amount $(b-x_H)g(x_H)/x_H$ equally to each among the $nx_H$ $H$-path participants to keep the budget balanced\footnote{
 Note that in our notation $g(x_H)$ is charged to the $(1-x_H)$ users in the  $L$-path flow, although we denote it as
 a function of the value $x_H$ of the $H$-path flow.}.

 Our side-payment mechanism only depends on a user's path choice and not on her private type, and will ensure IC.
This  side-payment mechanism requires a billing system (e.g., Amazon Mechanical Turk) to enable monetary transfer between users. As long as users want to access the public information, they need to accept this mechanism's policy and use the billing system. In Section IV, we will introduce another incentive mechanism based on content-restriction, that does not
require such a billing system.

 After introducing the side-payment mechanism, the payoffs of users change. Depending on her path choice ($H$ or $L$), a user of type $i\in\{1,2\}$ will receive a refund $(b-x_H)g(x_H)/x_H$ or have to submit payment $g(x_H)$. Her payoff changes from (\ref{payoff}) to
\begin{equation}
\footnotesize
\label{payoff2}
u_i(P,x_H)=\left\{
\begin{array}{ll}
\theta_iQ(x_H,b)-g(x_H)&\mbox{if }P=L;\\
\theta_iQ(x_H,b)-c_H+\frac{b-x_H}{x_H}g(x_H)&\mbox{if }P=H.
\end{array}
\right.
\end{equation}

Consider any positive equilibrium $\hat{x}_H>0$. A type-$i$ user's payoff difference
between taking the $H$ and $L$ paths is zero when
$$-c_H+\frac{b-\hat{x}_H}{\hat{x}_H}g(\hat{x}_H)+g(\hat{x}_H)=\frac{b}{\hat{x}_H}g(\hat{x}_H)-c_H\,,$$
which does not depend on its type. Hence in order to have a positive flow
on both paths we need the above incentive difference to be zero, i.e.,
the payment function $g$ must satisfy
\begin{equation}
\label{condition}
g(\hat{x}_H)=\hat{x}_H\frac{c_H}{b}\,.
\end{equation}

We say an equilibrium is asymptotically stable \cite{fudenberg1991game} if a small perturbation of the traffic on the two links, say adding
(subtracting) small $\epsilon$ to the $H$-path and subtracting (adding) the $\epsilon$ from $L$-path, does not move the system away from the equilibrium. Stability is an important criterion for an equilibrium to be implementable,
thus we only consider stable equilibria in this paper.
We say that \emph{we can incentivise (as an equilibrium) the flow value $\hat{x}_H$}
 if $\hat{x}_H$ can be obtained as a stable equilibrium for some payment function $g(x_H)$.
We now provide the design of a budget balanced payment scheme that incentivises any flow partition between the two paths
(without considering IR, which will be dealt later).
\begin{proposition}
\label{lemma1}
Assume BB but not necessarily IR, a total mass of users $b\leq1$, and any target equilibrium $\hat{x}_H$.
The following side-payment function $g_{\hat{x}_H}(x_H)$ incentivizes $\hat{x}_H$:
\begin{equation}
\label{function}
g_{\hat{x}_H}(x_H)=\frac{\hat{x}_Hc_H(b-x_H)}{b(b-\hat{x}_H)}\,.
\end{equation}
Observe that at the target equilibrium $\hat{x}_H$, each user's total perceived cost
(travel cost plus payment or minus subsidy) is equal to $\hat{x}_Hc_H/b$, independently of the path choice.
Thus the equilibrium payoff of a type-$i$ user is
 \begin{equation}
\label{payoff6}
\theta_iQ(\hat{x}_H,b)-g_{\hat{x}_H}(\hat{x}_H)=\theta_iQ(\hat{x}_H,b)-\frac{\hat{x}_H}{b}c_H\,.
\end{equation}
\end{proposition}

Notice that the payment function (\ref{function}) becomes (\ref{condition}) when $x_H=\hat{x}_H$, which ensures that $\hat{x}_H$ is an equilibrium under the side-payment design in (\ref{function}). (\ref{function}) also
implies convergence from any $x_H$ to the stable target $\hat{x}_H$:
\begin{itemize}
\item If $x_H>\hat{x}_H$, the corresponding (\ref{function}) is smaller than equilibrium payment in (\ref{condition}), motivating more users to choose $L$-path with smaller penalty and $x_H$ decreases.
\item If $x_H<\hat{x}_H$, (\ref{function}) is larger than (\ref{condition}) and we expect a larger $x_H$ for choosing $H$-path to get refund.
    \end{itemize}

At the stable equilibrium $\hat{x}_H$, each $L$-path participant perceives an extra cost $\hat{x}_Hc_H/b$ as in (\ref{condition}), where each user is indifferent in choosing between the two paths, but this cost may not be covered by her information value and can drive her out of the crowdsensing system. Hence, besides Proposition \ref{lemma1}, we need to consider IR in the side-payment design.
We say that an equilibrium $\hat{x}_H$ satisfies IR for any type-$i$ user if,
\begin{equation}
\label{IR}
\theta_iQ(\hat{x}_H,b)-\frac{\hat{x}_H}{b}c_H\ge0\,.
\end{equation}
In the following, we show that if users are homogeneous in their valuation types (i.e., $\theta_1=\theta_2$), IR always holds when using Proposition \ref{lemma1}; whereas if they are heterogeneous with $\theta_1<\theta_2$, IR does not hold in general.

\subsection{Side-payment for Homogeneous Users}\label{sidepaymenthomo}
Due to $\theta_1=\theta_2$ in this subsection, the average valuation $\theta_0$ equals $\theta_1$ or $\theta_2$. To avoid trivial discussion, we assume the optimal social welfare at the flow maximizer $x_H^*$ in (\ref{swl}) is nonnegative, i.e.,
{\setlength\abovedisplayskip{1pt}
\setlength\belowdisplayskip{1pt}
$$SW(x_H^*)=\theta_0Q(x_H^*,1)-x_H^*c_H\geq 0,$$
}which tells that  under full participation ($b=1$), any user's payoff in (\ref{IR}) is non-negative. Thus, IR always holds for the homogeneous user case and according to Proposition \ref{lemma1} we can ideally choose $\hat{x}_H=x_H^*$ in (\ref{swl}) to
inentivize the social optimum. We denote by $PoA_g$ the $PoA$ we can achieve using incentive payments.
\begin{proposition}
\label{theorem1}
 IR holds always for the homogeneous user case and we can
use the the side-payment mechanism $g(x_H)$ in (\ref{function}) to incentivize the social optimum $x_H^*$ in (\ref{swl}).
In this case  $PoA_g=1$.
\end{proposition}


\subsection{Side-payment for Heterogeneous Users}\label{subsection:heterogeneous}
Now we turn to a more general case with $\theta_1<\theta_2$ and
for simplicity of the analysis from now on we assume the users' proportions in both types are equal ($\eta=0.5$)\footnote{Our analysis and results can be easily extended to other values of $\eta$.}.
This $\eta = 0.5$ value is assumed to be common knowledge, obtained in practice by obtaining users' type statistics.
Then we have $\theta_0=(\theta_1+\theta_2)/2$ with $\theta_1<\theta_0<\theta_2$. In this heterogeneous case,
$SW(x_H^*)\geq 0$ no longer guarantees IR at the desired $x_H^*$ for the lower type-1 users,
and we cannot apply Proposition \ref{lemma1} directly.
\begin{figure}[!t]
  \centering
  \includegraphics[width=3in]{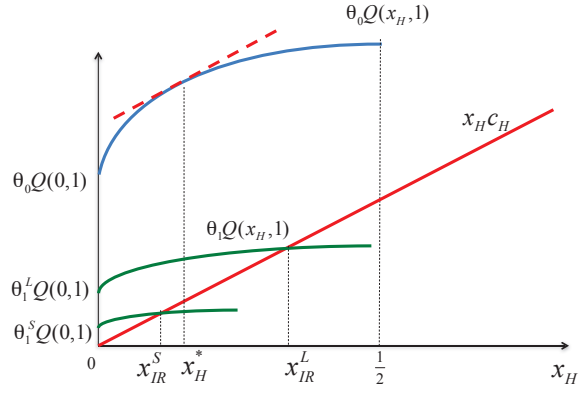}
  \caption{The IR constraint in the SW optimization problem. The free SW maximization chooses $x_H^*$. When $\theta_1$ is small (e.g., $\theta_1^S$ here), the IR constraint becomes binding since the corresponding small $x_{IR}^S$ is less than $x_H^*$. For large $\theta_1^L$ the IR constraint is relaxed with $x_{IR}^L\geq x_H^*$. Note that $x_{IR}$ moves to the left as $c_H$ increases or as $\theta_1$ decreases.}
  \label{figure1}
\end{figure}
The IR condition for type-1 is (see \eqref{IR} with $b=1$)
\begin{equation}
\label{cons}
\theta_1Q(x_H,1)-x_Hc_H\geq0\,,
\end{equation}
indicating that if we like to keep type-1 users in the system we must
restrict our targeted equilibria in the range $[0,x_{IR}]$,
which is the unique solution to \eqref{cons} holding with equality.
Fig. \ref{figure1} plots the user-average information value $\theta_0Q(x_H,1)$
and the type-1 user's information value $\theta_1Q(x_H,1)$ as functions of $x_H$,
 for a large and a small value of $\theta_1$ values (we can choose any $\t_1 < \t_0$ and take
 $\t_2 =2\t_0 - \t_1$).
 We observe that the social optimum $x_H^*$
 of the free optimisation problem (no IR consideration for type-1) is achieved when the gap between $\theta_0Q(x_H,1)$
 and total cost $x_Hc_H$ is maximized as in (\ref{swl}).
 Whether we can achieve $x_H^*$ in the constraint problem (with IR for type-1) depends on the value of $\t_1$:
 \begin{itemize}
   \item When $\theta_1$ is large (e.g., $\theta_1=\theta_1^L$ in Fig. \ref{figure1}), the corresponding $x_{IR}=x_{IR}^{L}$ satisfies $x_{IR}^L > x_H^*$ in Fig. \ref{figure1},
   and  IR is satisfied at the social optimum $x_H^*$. In this case we apply (\ref{function}) to achieve $SW(x_H^*)$.
   \item When $\theta_1$ is small (e.g., $\theta_1=\theta_1^S$ in Fig. \ref{figure1}), we have $x_{IR}=x_{IR}^S$ and $x_{IR}^S < x_H^*$ in Fig. \ref{figure1},
   and  IR is not satisfied at the social optimum $x_H^*$.
   The best equilibrium decision $x_H^*$ is to either operate the system at  $x_{IR}^S$ and ensure both types' IR
   (the so-called \emph{full participation case} with $b=1$), or to incentivise a flow $x_H>x_{IR}^S$
   that maximises the efficiency in a system where only type-2 users participate
   (in the so-called \emph{half participation case} with $b=0.5$).
 \end{itemize}

How to maximize social welfare under our side-payment mechanism is now clear. First, we let $\tilde{x}_H$ be the solution to the free
social welfare maximisation problem in the half participation case (only type-2 participate) with $b=0.5$:
\begin{equation}
\label{half}
\tilde{x}_H\in\argmax_{x_H\in[0,0.5]}\{0.5\theta_2Q(x_H,0.5)-x_Hc_H\}\,.
\end{equation}
\begin{figure}[!t]
  \centering
  \includegraphics[width=3in]{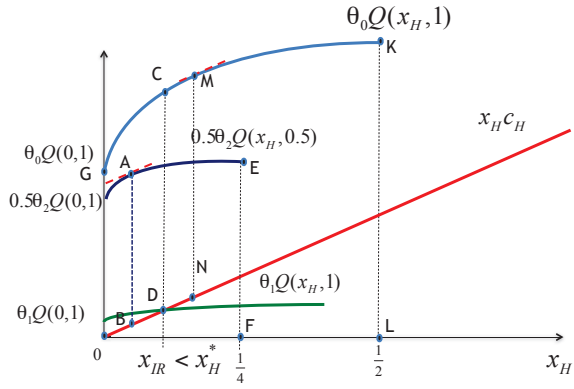}
  \caption{The solution of SW optimization problem where both types or only type-2 participates. $MN$ with $x_H^*$ is the social optimum and $AB$ with $\tilde{x}_H$ is the optimum value when we only include type-2 only. $CD$ with $x_{IR}$ is the solution if we do not violate IR for type-1 and include both types.}
  \label{figure2}
\end{figure}

 Figure \ref{figure2} shows the solutions of our three optimisation problems: the free optimisation problem, the
case where only type-2 participate (half participation),
and the best we can do if we insist in keeping both types in the system (full participation).
By comparing the optimum social welfare we can achieve in these three cases we obtain the following
procedure for choosing the equilibirum:
\begin{theorem}
\label{prop3}
In our heterogeneous user model, the equilibrium flow partition $\hat{x}_H$
and the resultant optimal social welfare $SW_g$
are chosen as follows:
\begin{itemize}
  \item \emph{Social optimum}: If $x_H^*\le x_{IR}$, then we optimally choose $\hat{x}_H=x_H^*$ and $SW_g=SW(x_H^*)$. (See $MN$ in Fig. 3).
  \item \emph{Full participation}: If $x_H^*> x_{IR}$ and $SW(x_{IR})\ge SW(\tilde{x}_H)$, then we optimally choose $\hat{x}_H=x_{IR}$ and $SW_g=SW(x_{IR})$. (See $CD$ in Fig. 3).
  \item \emph{Half participation}: If $x_H^*> x_{IR}$ and $SW(x_{IR})<SW(\tilde{x}_H)$, we optimally choose $\hat{x}_H=\tilde{x}_H$ and $SW_g=SW(\tilde{x}_H)$. (See $AB$ in Fig. 3).
\end{itemize}
By searching over all possible parameters and information value function $Q_1(\cdot)$, the side-payment mechanism for two user types achieves $PoA_g=1/2$.
\end{theorem}
\begin{proof}
We have to compare the full participation case with $b=1$ and the half participation case with $b=0.5$ to optimally decide $\hat{x}_H$. If $x_H^*\le x_{IR}$ then our optimum of the full participation case with $b=1$ is clearly at $x_H^*$ which  is the global optimum. Then we don't need to consider the half participation case with $b=0.5$.

If $x_H^*> x_{IR}$, then we must consider the half participation case with $b=0.5$. This is because we may perform even better by operating a single type $\theta_2$ system in the range $[x_{IR},0.5]$. In the half participation case with $b=0.5$, we operate the type $\theta_2$ system alone since type $\theta_1$ refuse to participate. The best we can do in this case is $SW^\prime\le SW(\tilde{x}_H)$ since it could be that the optimal $\tilde{x}_H$ is $\tilde{x}_H<x_{IR}$ and not in the allowable range $[x_{IR},0.5]$. Assume $SW(x_{IR})> SW(\tilde{x}_H)$, then when $\tilde{x}_H<x_{IR}$, $SW(x_{IR})> SW(\tilde{x}_H)\ge SW^\prime$; when $\tilde{x}_H\ge x_{IR}$, $SW(x_{IR})> SW(\tilde{x}_H)=SW^\prime$. Thus when $SW(x_{IR})> SW(\tilde{x}_H)$, it is optimal to keep both types in the system. Assume finally that $SW(x_{IR})\le SW(\tilde{x}_H)$, then it cannot be the case that $\tilde{x}_H<x_{IR}$ which will be shown next, implying that $SW_g=SW(\tilde{x}_H)=SW^\prime$. Thus when $SW(x_{IR})\le SW(\tilde{x}_H)$, it is optimal to keep only type-2 in the system.

Now we need to prove our claim that $\tilde{x}_H\ge x_{IR}$ when $SW(x_{IR})\le SW(\tilde{x}_H)$. Remember that we are in the region $x_{IR}<x_H^*$. Assume that $\tilde{x}_H<x_{IR}$, and $SW(\tilde{x}_H)>SW(x_{IR})$. We will get a contradiction. At $\tilde{x}_H$, clearly $0.5\theta_2Q(\tilde{x}_H,0.5)\le\theta_0Q(\tilde{x}_H,1)$, hence $SW(\tilde{x}_H)=0.5\theta_2Q(\tilde{x}_H,0.5)-\tilde{x}_Hc_H\le\theta_0Q(\tilde{x}_H,1)-\tilde{x}_Hc_H$. Also because of the concavity of $Q(x_H,1)$, $\theta_0Q(x_H,1)-x_Hc_H$ is increasing for $x\le x_H^*$, hence $\theta_0Q(\tilde{x}_H,1)-\tilde{x}_Hc_H\le \theta_0Q(x_{IR},1)-x_{IR}c_H=SW(x_{IR})$. But this implies that $SW(\tilde{x}_H)\le SW(x_{IR})$, a contradiction.

Next we prove $POA_g=1/2$. Note that, if $ x_H^*<x_{IR}$,
$$SW_g=SW( x_H^*)\ge\theta_0Q(0,1),$$
and if $ x_H^*\ge x_{IR}$,
$$SW_g\ge SW(x_{IR})\ge\theta_0Q(0,1).$$
Thus,
\begin{eqnarray}
\label{poa}
PoA_g&=&\min\limits_{\theta_i,c_H,Q(\cdot,1)}\frac{SW_g}{SW( x_H^*)}\ge\min\limits_{\theta_i,Q(\cdot,1)}\frac{\theta_0Q(0,1)}{SW( x_H^*)}\nonumber\\
&\ge&\min\limits_{Q(\cdot,1)}\frac{Q(0,1)}{Q(0.5,1)}=\min\limits_{Q_1(\cdot)}\frac{Q_1(1)}{2Q_1(0.5)}\ge\frac{1}{2}\nonumber
\end{eqnarray}

We next show that the bound is tight. Define $Q_1(x)$ by (\ref{fucdef}), and consider any fixed value of $c_H$ that is less than $2\theta_0q$. We consider what happens as $\theta_1\rightarrow0$. One can easily check the following:
\begin{itemize}
  \item[1.] The free social optimum: The optimum is attained for $x_H=0.5$, $SW( x_H^*)=2\theta_0q-0.5c_H$.
  \item[2.] The optimal solution of the full participation case with $b=1$: The IR condition for type $\theta_1$ is $\theta_1(q+2qx_H)-x_Hc_H\ge0$. Solving this we obtain
       $$x_{IR}=\frac{\theta_1q}{c_H-2\theta_1q}$$
        which converges to 0 when $\theta_1\rightarrow0$. Hence, as $\theta_1\rightarrow0$ the optimum value $SW(x_{IR})\rightarrow\theta_0Q(0,1)=\theta_0q$.
  \item[3.] The optimal solution of the half participation case with $b=0.5$: Since for our constructed function $Q(x_H,0.5)=q$ for $x_H\in[0,0.5]$. The optimum is achieved at $x=0$ and $SW(\tilde{x}_H)=0.5\theta_2Q(0,0.5)=0.5\theta_2q$ which is $\theta_0q$ in the limit as $\theta_1\rightarrow0$.
\end{itemize}
The $PoA_g$ for this specific instance becomes
$$PoA_g=\min\limits_{c_H}\frac{\max\{SW(x_{IR}),SW(\tilde{x}_H)\}}{SW( x_H^*)}$$
$$=\min\limits_{c_H}\frac{\theta_0q}{2\theta_0q-0.5c_H}=\frac{1}{2}.$$
This completes the proof.
\end{proof}

Unlike $PoA_g=1$ in the homogeneous case, here in the worst case our incentive mechanism can be as bad as the original system without incentives
 (as in Proposition \ref{prop2}). Intuitively, when users are diverse,
 we have to decide to selectively include less users (exclude low valuation users by raising the incentive payment) to achieve better path diversity,
 or to keep  a larger number of users participating at the cost of path diversity (by lowering the incentive payment
 and hence weakening the incentives to use the $H$-path).
 In the half participation case ($b=0.5$), we miss the low-valuation users' contribution in sensing. In the full participation case ($b=1$),
 we reduce path diversity and mostly collect one-path information.
 If path diversity and full participation are both critical to achieving the optimal social welfare, we may lose half of the optimum as indicated by Theorem \ref{prop3}. In Section~\ref{simulation}, we will choose $Q_1(\cdot)$ as (\ref{Q}) and use extensive simulations to show that this mechanism is still efficient in most of the time.

 It goes beyond the scope of this paper to discuss issues related to the form of the content function $Q(x,1)$. For instance, if our unit mass of users
corresponds to a very large value $n$ of actual users, even a small fraction $x$ on a path can obtain a substantial amount of the total content
and path diversity is important.
If $n$ is very small, we expect path diversity to play a less role, since these few users will discover new content even if they travel on the same path.

\section{Content-restriction as Incentive}\label{contentrestriction}
As we have already discussed, using side-payment as an incentive mechanism needs monetary transfer among users, which may not be possible in some applications. This motivates us to consider a payment-free incentive design: as a penalty  to $L$-path participants, restrict the amount of content/information made available to them. Content restriction in the information sharing system can be done directly by making some information invisible or indirectly by offering a lower quality application service to the given user.

 \emph{Content-restriction Mechanism:} To motivate more users to choose the $H$-path,
the system planner provides a fraction $a$ of the total information, i.e,  $aQ(x_H,1)$,
 to $L$-path participants, where $a\in[0,1]$ is the content restriction coefficient.
 For example, if we choose $a=1$ (or 0), the system offers full (or zero) information to $L$-path participants.

 Our content-restriction mechanism\,only\,depends\,on\,a\,user's path choice and is independent of her type. It satisfies the IC and IR requirements
(since users can always take the $L$-path at zero cost, we have always full participation).
Under the mechanism with coefficient $a$, a type-$i$ user's payoff changes from (\ref{payoff}) to
\[
\label{payoff3}
u_i(P,x_H)=\left\{
\begin{array}{ll}
a\theta_iQ(x_H,1), &\mbox{if }P=L;\\
\theta_iQ(x_H,1)-c_H, &\mbox{if }P=H.
\end{array}
\right.
\]
If $u_i(H,x_H)\ge u_i(L,x_H)$ or simply
\begin{equation}
\label{content2}
Q(x,1)\ge\frac{c_H}{(1-a)\theta_i},
\end{equation}
she will choose $H$-path. Otherwise, she will choose $L$-path.
As $a$ decreases, we purposely ``destruct" more shared information to $L$-path participants
and condition (\ref{content2}) for choosing $H$-path is more likely to hold
(even more for the high user types).
\subsection{Content-restriction for Homogeneous Users}
Let us first consider the homogenous case with $\theta_i=\theta_0$ for $i=1,2$. If we expect positive flows with $\hat{x}_H\in (0,1)$ on both paths at the equilibrium, all users should be indifferent in choosing between the two paths. Users' payoffs by choosing $L$- and $H$-path are equal, i.e.,
\begin{equation}
\label{condition4}
Q(\hat{x}_H,1)=\frac{c_H}{(1-a)\theta_0}.
\end{equation}

In Fig. \ref{figure3}, we plot $Q(x_H,1)$ as a function of $x_H$ and compare to the path decision threshold $c_H/((1-a)\theta_0)$ in (\ref{condition4}). An intersection point (if any) of these two curves tells the equilibrium where the users are indifferent in choosing between  the two paths. The threshold $c_H/((1-a)\theta_0)$ increases with $a$ and can have at most two interaction points with $Q(x_H,1)$. To fully characterize different intersection or equilibrium results, we define two non-negative thresholds for $a$:
\begin{equation}
\label{underlinea}
\underline{a}=\max\{0,\frac{\theta_0\underline{Q}-c_H}{\theta_0\underline{Q}}\},\ \mbox{where}\ \underline{Q}=Q(0,1)\,,
\end{equation}
\begin{equation}
\label{overlinea}
\overline{a}=\max\{0,\frac{\theta_0\overline{Q}-c_H}{\theta_0\overline{Q}}\},\ \mbox{where}\ \overline{Q}=Q(0.5,1)\,.
\end{equation}
Such thresholds help define the following three regimes:
\begin{itemize}
  \item \emph{Strong content-restriction regime} ($0\le a<\underline{a}$): When $a$ is small and there is much content destruction as penalty to $L$-path participants, $c_H/((1-a)\theta_0)$ is small and does not intersect with $Q(x_H,1)$, and the only stable equilibrium is that all users choose $H$-path (i.e., $\hat{x}_H=1$).
  The social welfare is $SW=\theta_0\underline{Q}-c_H$ which is the same for any $a\in[0,\underline{a})$.
  \item \emph{Medium content-restriction regime} ($\underline{a}\le a<\overline{a}$): When $a$ is medium, $c_H/((1-a)\theta_0$ in Fig. \ref{figure3} intersects with $Q(x,1)$ at two points that are equilibria:
  one is $x_{H0}\in(0.5,1]$ and the other is $1-x_{H0}$. We can show that only $x_{H0}$ is stable\footnote{$1-x_{H0}$ is not stable after some perturbation. For example, reducing $x_H^*$ slightly encourages more $H$-path participants to churn to $L$-path. Similarly, equilibrium  $x_{H0}$ is stable: decrease from $x_H^*$ encourages more $L$-path participants to churn to $H$-path, and increase from $x_H^*$ encourages more $H$-path participants to churn to $L$-path.}. In this regime both paths offer the same payoffs to the users, hence each user obtains a payoff of $\t_0Q(x_{H0},1)-c_H$. Figure \ref{figure3} also shows that as $a$ increases, $x_{H0}$ moves towards $0.5$,
    and the corresponding social welfare $\theta_0Q(x_{H0},1)-c_H$ increases towards the maximum $\theta_0\overline{Q}-c_H$. Note that $x_H=0.5$ is not a stable equilibrium, and hence we like to get arbitrarily close:  choose $a=\bar{a}-\epsilon$ with infinitesimal $\epsilon>0$ to reach the social optimum $SW=\theta_0\overline{Q}-c_H$ asymptotically. 
  \item \emph{Weak content-restriction regime} ($\overline{a}\le a\le1$): When $a$ is large, the only stable equilibrium is that all users choose $L$-path. The corresponding social welfare is $SW=\theta_0a\underline{Q}$ and it is maximised by deciding minimum destruction ($a=1$).
\end{itemize}
\begin{figure}[!t]
  \centering
  \includegraphics[width=2.5in]{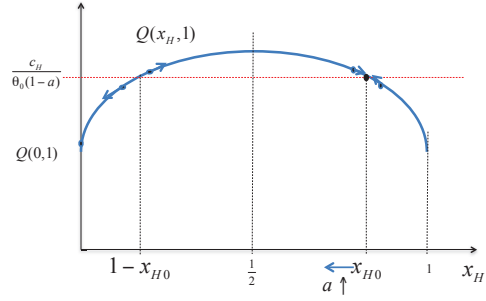}
  \caption{Intersection of $Q(x_H,1)$ and $c_H/(\theta_0(1-a))$ to decide equilibrium $x_H^*$. From the two possible equilibria ($x_{H0}$ and $1-x_{H0}$), only $x_{H0}>0.5$ is stable, where more traffic is routed towards the H-path. As $a$ increases, both equilibria tend to move towards $0.5$.}
  \label{figure3}\vspace{-10pt}
\end{figure}

Strong or weak information destruction on $L$-path participants causes no path-diversity at the equilibrium. Only in the medium destruction regime we reach perfect path-diversity ($\hat{x}_H\rightarrow0.5$) by choosing $a=\overline{a}-\epsilon$.
When $c_H$ is small, the benefit gained due to path diversity in the medium destruction regime covers the  efficiency loss due to
 content destruction. But when $c_H$ is large, content destruction and hence efficiency loss becomes very large
 if one likes to incentivize path-diversity, motivating the use of $a=1$.
\begin{theorem}
\label{prop7}
Let $SW_a$ be the maximum social welfare by applying the content-restriction.
\begin{itemize}
  \item If the travel cost over $H$-path is small (i.e., $c_H<\theta_0(\overline{Q}-\underline{Q})$), we optimally choose $a=\overline{a}-\epsilon$ to
  approach perfect path diversity ($\hat{x}_H\rightarrow0.5$) and optimum social welfare ($SW_a\rightarrow\theta_0\overline{Q}-c_H$).
  \item If the travel cost is large ($c_H\ge\theta_0(\overline{Q}-\underline{Q})$), it is optimal to choose $a=1$ since
  content destruction must be too excessive to incentivize path diversity.
  The corresponding social welfare is $SW_a=\theta_0\underline{Q}$
\end{itemize}
By searching over all possible parameters and information value function $Q_1(\cdot)$, the content-restriction mechanism for uniform user types achieves $PoA_a=1/2$.
\end{theorem}
\begin{proof}
In the medium content-restriction regime, we have shown that the maximal social welfare attained at a stable path-diversity equilibrium ($0<\hat{x}_H<1$) is $\theta_0\overline{Q}-c_H$ which is reached asymptotically. In the strong and weak content-restriction regime, we can infer that the maximal social welfare attained at a stable zero path-diversity equilibrium ($\hat{x}_H=0$ or $1$) is $\theta_0\underline{Q}$. Thus when  $c_H<\theta_0(\overline{Q}-\underline{Q})$, we have
$$\theta_0\overline{Q}-c_H>\theta_0\underline{Q}.$$
Then we optimally choose $a=\overline{a}-\epsilon$\footnote{Note that this value of $a$ is positive since $c_H<\theta_0\overline{Q}$.} to approach perfect path diversity ($\hat{x}_H\rightarrow0.5$) and optimum social welfare $SW_a\rightarrow\theta_0\overline{Q}-c_H$. When $c_H\ge\theta_0(\overline{Q}-\underline{Q})$, we have
$$\theta_0\overline{Q}-c_H\le\theta_0\underline{Q}.$$
Then we optimally choose $a=1$ to reach optimum social welfare $SW_a=\theta_0\underline{Q}$.

Next we prove $PoA_a=1/2$. We denote the set of all the equilibria when we apply content restriction with coefficient $a$ by $NE(a)$, for example, $NE(1)=\{0\}$ means when $a=1$ there exists only one equilibrium which is all users choose the low cost path. Based on this notation we have
$$PoA_a=\min\limits_{\theta_i,Q(\cdot,1),c_H}\frac{\max\limits_{a}\min\limits_{\hat{x}_H\in NE(a)}\{SW(\hat{x}_H)\}}{SW(\hat{x}_H)}.$$
We claim that among all the equilibria the full content-preserving equilibrium $a=1$ maximises the worst case equilibrium for all values of $c_H$. Now we prove our claim. When $a=1$, $\hat{x}_H=0$ is the unique equilibrium and the corresponding social welfare is $\theta_0\underline{Q}$. Assume that for some other $a^\prime$, the worse equilibrium has efficiency higher than that. Then this equilibrium must be either one where users are indifferent between the two paths, or one where all users choose the same path. In the first case, we know that there is a second equilibrium where customers select the the low cost path with efficiency $a^\prime\theta_0\underline{Q}$  and this must be superior to the one where users are indifferent. But leads to contradiction since it would imply $a^\prime\theta_0\underline{Q}>\theta_0\underline{Q}$ while $a^\prime<1$. Suppose that the best equilibrium corresponds to users not being indifferent. But then for all values of $0\le a^\prime\le1$ this cannot be superior to the one with $a=1$ we proposed.

We just proved that $\max\limits_{a}\min\limits_{\hat{x}_H\in NE(a)}\{SW(\hat{x}_H)\}=\theta_0\underline{Q}$, then
$$PoA_a=\min\limits_{\theta_i,Q(\cdot,1),c_H}\frac{\theta_0\underline{Q}}{SW(\hat{x}_H)}\ge\frac{\underline{Q}}{\overline{Q}}\ge\frac{1}{2}.$$
We next show that the bound is tight. Define $Q_1(x)$ by (\ref{fucdef}), and consider any fixed value of $c_H$ that is less than $2\theta_0q$. The $PoA_a$ for this specific instance becomes
$$PoA_a=\min\limits_{\theta_i,Q(\cdot,1),c_H}\frac{\theta_0\underline{Q}}{SW(\hat{x}_H)}=\min\limits_{c_H}\frac{\theta_0q}{2\theta_0q-0.5c_H}=\frac{1}{2}.$$
Therefore, $PoA_a=1/2$.
\end{proof}

Using content-restriction as incentive, we either accept the original zero path-diversity equilibrium or perfect path-diversity in information collection. Unless $c_H$ is so high that we always prefer a zero diversity equilibrium, we can achieve perfect path-diversity by content restriction. Though this mechanism is not efficient when $c_H$ is large, we will see that user diversity helps improve it in the next subsection.

\subsection{Content-restriction for Heterogenous Users}
Now we turn to the more general case with $\theta_1<\theta_2$. In this heterogeneous case, we search for an $a\in[0,1]$ to maximise social welfare at the equilibrium. If $a$ is small or large, we only have zero path-diversity equilibria with either $\hat{x}_H=1$ or 0. As long as the $H$-path cost $c_H$ is not too large, we should optimally choose a medium value of $a$ to incentivise the perfect path-diversity equilibrium ($\hat{x}_H=0.5$).

If we expect positive flows with $\hat{x}_H\in(0,1)$ on both path at the equilibrium, we need to make either type-1 or type-2 be indifferent between the two paths. In Fig. \ref{figure4}, we plot $Q(x_H,1)$ and $c_H/(\theta_i(1-a)),\ i=1,2$ as functions of $x_H$. If we make type-1 users be indifferent between the two paths, then
$$Q(\hat{x}_H,1)=\frac{c_H}{(1-a)\theta_1}.$$
Since $\theta_2>\theta_1$, we have
$$Q(\hat{x}_H,1)>\frac{c_H}{(1-a)\theta_2},$$
we infer that type-2 users will choose $H$-path and the corresponding equilibrium $\hat{x}_H\in[0.5,1]$. If we make type-2 users be indifferent between the two paths, then
$$Q(\hat{x}_H,1)=\frac{c_H}{(1-a)\theta_2}.$$
Since $\theta_2>\theta_1$, we have
$$Q(\hat{x}_H,1)<\frac{c_H}{(1-a)\theta_1},$$
we infer that type-1 users will choose $L$-path and the corresponding equilibrium $\hat{x}_H\in[0,0.5]$.
\begin{figure}[!t]
  \centering
  \includegraphics[width=2.5in]{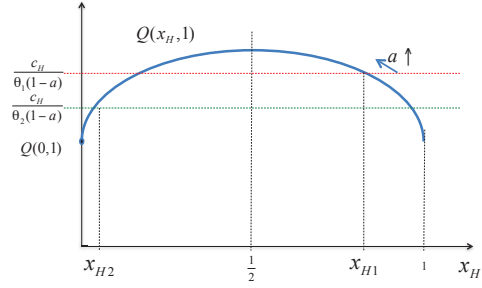}
  \caption{Intersection of $Q(x_H,1)$ and $c_H/(\theta_i(1-a)),\ i=1,2$ to decide the equilibrium $x_H^*$. Only $x_{H1}$ and $x_{H2}$ in the figure are possible equilibria and $x_{H1}$ is stable, $x_{H2}$ is unstable.}
  \label{figure4}
\end{figure}

In Fig. \ref{figure4},  we plot $Q(x_H,1)$ as a function of $x_H$ and compare to the path decision threshold $c_H/((1-a)\theta_1)$ and $c_H/((1-a)\theta_2)$. An the intersection point (if any) of $c_H/(\theta_i(1-a))$ and $Q(x_H,1)$ tells the possible equilibrium where the users are indifferent in choosing between the two paths. In Fig. \ref{figure4}, only $x_{H1}$ and $x_{H2}$ in the figure are possible equilibria. To fully characterize different intersection or equilibrium results, we define four non-negative thresholds for $a$:
\begin{equation}
\label{underlinea}
\underline{a}_i=\max\{0,\frac{\theta_i\underline{Q}-c_H}{\theta_i\underline{Q}}\},\ \mbox{where}\ \underline{Q}=Q(0,1)\ \,,
\end{equation}
\begin{equation}
\label{overlinea}
\overline{a}_i=\max\{0,\frac{\theta_i\overline{Q}-c_H}{\theta_i\overline{Q}}\},\ \mbox{where}\ \overline{Q}=Q(0.5,1)\,.
\end{equation}
First, we assume $\theta_2/\theta_1>\overline{Q}/\underline{Q}$, i.e., $\overline{a}_1<\underline{a}_2$. This condition ensures that the decision thresholds are separated in a sense that when the decision threshold for type-1 is below $\overline{Q}$, the information threshold for type-2 must be below $\underline{Q}$ and when the information threshold for type-2 is beyond $\underline{Q}$, the information threshold for type-1 must be beyond $\overline{Q}$. Clearly, we have $\underline{a}_1<\overline{a}_1<\underline{a}_2<\overline{a}_2<1$.
Such thresholds help define the following four regimes:
\begin{itemize}
\item \emph{Strong content restriction regime} ($0\le a<\underline{a}_1$): When $a$ is small and there is much content destruction as penalty to $L$-path participants, both $c_H/((1-a)\theta_1)$ and $c_H/((1-a)\theta_2)$  is small and does not intersect with $Q(x_H,1)$, and the only stable equilibrium is that all users choose $H$-path (i.e., $\hat{x}_H=1$).
     The social welfare is $SW=\theta_0\underline{Q}-c_H$ which is the same for any $a\in[0,\underline{a})$.
  \item \emph{Medium content-restriction regime} ($\underline{a}_1\le a\le\overline{a}_1$): In this regime, $c_H/((1-a)\theta_2)$ in Fig. \ref{figure4} is below $Q(x_H,1)$ and $c_H/((1-a)\theta_1$ intersects with $Q(x_H,1)$ at two points. All type-2 users will choose $H$-path and only the right intersection point $x_{H1}\in[0.5,1]$ is an positive path-diversity equilibrium. We can show that $x_{H1}$ is stable\footnote{Equilibrium  $x_{H1}$ is stable: decrease from $x_{H1}$ encourages more $L$-path participants to churn to $H$-path, and increase from $x_{H1}$ encourages more $H$-path participants to churn to $L$-path.}. Figure \ref{figure4} also shows that as $a$ increases, $x_{H1}$ moves towards $0.5$, and the corresponding social welfare increases towards the maximum $\theta_0\overline{Q}-c_H$. Note that $x_H=0.5$ is a stable equilibrium, and hence we choose $a=\overline{a}_1$ to reach the social optimum $SW=\theta_0\overline{Q}-c_H$.
  \item \emph{Lower medium content restriction regime} ($\overline{a}_1< a<\underline{a}_2$):
  In this regime, $c_H/((1-a)\theta_2)$ in Fig. \ref{figure4} is below $Q(x_H,1)$ and $c_H/((1-a)\theta_1$ is above $Q(x_H,1)$.  Thus all type-1 users will choose $L$-path and all type-2 users will choose $H$-path. Thus the only equilibrium is 0.5 and it is stable. Thus we choose $a=\underline{a}_2-\epsilon$ with infinitesimal $\epsilon>0$  to reach the social optimum $$\theta_0\underline{a}_2\overline{Q}=\theta_0\overline{Q}-\frac{\theta_1\overline{Q}+\theta_2\underline{Q}}{2\theta_2\underline{Q}}c_H.$$
\item \emph{Weak content restriction regime} ($\underline{a}_2< a\le1$): When $a$ is large, $c_H/((1-a)\theta_1)$ in Fig. \ref{figure4} is above $Q(x_H,1)$ and $c_H/((1-a)\theta_2$ intersects with $Q(x_H,1)$ at two points. All type-1 users will choose $L$-path and only the left intersection point $x_{H2}\in[0,0.5]$ is an positive path-diversity equilibrium. But it is not stable\footnote{Equilibrium  $x_{H2}$ is not stable: decrease from $x_{H2}$ encourages more $H$-path participants to churn to $L$-path, and increase from $x_{H2}$ encourages more $L$-path participants to churn to $H$-path}.
    The only stable equilibrium in this regime is that all users choose $L$-path and the corresponding social welfare is $SW=\theta_0a\underline{Q}$. We choose $a=1$ to reach the social optimum $\theta_0\underline{Q}$.
\end{itemize}

Next, we assume $\theta_2/\theta_1\le\overline{Q}/\underline{Q}$, i.e., $\overline{a}_1\le\underline{a}_2$, and we have the following three regimes
\begin{itemize}
\item \emph{Strong content restriction regime} ($0\le a<\underline{a}_1$): When $a$ is small and there is much content destruction as penalty to $L$-path participants, both $c_H/((1-a)\theta_1)$ and $c_H/((1-a)\theta_2)$  is small and does not intersect with $Q(x_H,1)$, and the only stable equilibrium is that all users choose $H$-path (i.e., $\hat{x}_H=1$). The social welfare is $SW=\theta_0\underline{Q}-c_H$ which is the same for any $a\in[0,\underline{a})$.
  \item \emph{Medium content-restriction regime} ($\underline{a}_1\le a\le\overline{a}_1$): In this regime, $c_H/((1-a)\theta_2)$ in Fig. \ref{figure4} may intersect with $Q(x_H,1)$ at two points and $c_H/((1-a)\theta_1$ intersects with $Q(x_H,1)$ at two points. The right intersection point of $c_H/((1-a)\theta_2)$ and $Q(x_H,1)$ $x_{H1}\in[0.5,1]$ is an positive path-diversity equilibrium. The left intersection point of $c_H/((1-a)\theta_1)$ and $Q(x_H,1)$ $x_{H1}\in[0.5,1]$ (if exists) is another positive path-diversity equilibrium. However, only $x_{H1}$ is stable. Figure \ref{figure4} also shows that as $a$ increases, $x_{H1}$ moves towards $0.5$, and the corresponding social welfare increases towards the maximum $\theta_0\overline{Q}-c_H$. Note that $x_H=0.5$ is a stable equilibrium, and hence we choose $a=\overline{a}_1$ to reach the social optimum $SW=\theta_0\overline{Q}-c_H$.
\item \emph{Weak content restriction regime} ($\overline{a}_1< a\le1$): When $a$ is large, $c_H/((1-a)\theta_1)$ in Fig. \ref{figure4} is above $Q(x_H,1)$ and $c_H/((1-a)\theta_2$ may intersect with $Q(x_H,1)$ at two points. All type-1 users will choose $L$-path and only the left intersection point $x_{H2}\in[0,0.5]$ (if exists) is an positive path-diversity equilibrium. But it is not stable. The only stable equilibrium in this regime is that all users choose $L$-path and the corresponding social welfare is $SW=\theta_0a\underline{Q}$. We choose $a=1$ to reach the social optimum $\theta_0\underline{Q}$.
\end{itemize}

Similar to the homogeneous case, when the travel cost $c_H$ is too high, we may not choose the path-diversity, as it is prohibitively expensive for any user type to cover the $H$-path. The following theorem shows how to design an $a$ to achieve maximum social welfare. Let $SW_a$ denote the maximum social welfare achieved by our content-restriction mechanism by optimizing over $a$.
\begin{theorem}
\label{prop10}
When user types are diverse (i.e., $\theta_2/\theta_1>\overline{Q}/\underline{Q}$ with $\underline{Q}$ and $\overline{Q}$ given in (\ref{underlinea}) and (\ref{overlinea})), the optimal $a$ depends on $c_H$:
\begin{itemize}
\item If the travel cost on $H$-path is small, i.e., $$c_H<\frac{(\theta_1+\theta_2)(\overline{Q}-\underline{Q})\theta_2\underline{Q}}{\theta_1\overline{Q}+\theta_2\underline{Q}},$$
    we choose
     $a=(\theta_2\underline{Q}-c_H)/(\theta_2\underline{Q})-\epsilon$ with infinitesimal $\epsilon>0$ to reach the perfect path-diversity ($\hat{x}_H=0.5$) and social welfare, $$SW_a\rightarrow\theta_0\overline{Q}-\frac{\theta_1\overline{Q}+\theta_2\underline{Q}}{2\theta_2\underline{Q}}c_H.$$
\item Otherwise, the travel cost on $H$-path is too high to send any user there $(\hat{x}_H=0)$, and we choose
     $a=1$ to avoid any content destruction. The corresponding social welfare is $SW_a=\theta_0\underline{Q}$.
    \end{itemize}
    When user types are similar (i.e.,  $\theta_2/\theta_1\le\overline{Q}/\underline{Q}$), the choice of $a$ also depends on $c_H$:
\begin{itemize}
  \item If the travel cost on $H$-path is small (i.e., $c_H<\theta_0(\overline{Q}-\underline{Q})$), we optimally decide $a=(\theta_1\overline{Q}-c_H)/(\theta_1\overline{Q})$ to reach perfect path-diversity ($\hat{x}_H=0.5$). The corresponding optimal social welfare is $SW_a=\theta_0\overline{Q}-c_H$.
  \item Otherwise, the travel cost on $H$-path is too high to send any user there $(\hat{x}_H=0)$ and we choose $a=1$ to avoid large content destruction. The corresponding social welfare is $SW_a=\theta_0\underline{Q}$.
\end{itemize}
The resultant price of anarchy under the optimal content-restriction mechanism is $PoA_a=1/2$.
\end{theorem}
\begin{proof}
When $\theta_2/\theta_1>\overline{Q}/\underline{Q}$, in the corresponding medium and lower medium content-restriction regime we have shown that the maximal social welfare attained at a stable path-diversity equilibrium ($0<\hat{x}_H<1$) is
$$\theta_0\overline{Q}-\frac{\theta_1\overline{Q}+\theta_2\underline{Q}}{2\theta_2\underline{Q}}c_H,$$
which is reached asymptotically. In the corresponding strong and weak content-restriction regime, we can infer that the maximal social welfare attained at a stable zero path-diversity equilibrium ($\hat{x}_H=0$ or $1$) is $\theta_0\underline{Q}$. Thus if
$$c_H<\frac{(\theta_1+\theta_2)(\overline{Q}-\underline{Q})\theta_2\underline{Q}}{\theta_1\overline{Q}+\theta_2\underline{Q}},$$
we have
$$\theta_0\overline{Q}-\frac{\theta_1\overline{Q}+\theta_2\underline{Q}}{2\theta_2\underline{Q}}c_H>\theta_0\underline{Q}.$$
Then we optimally choose $a=\underline{a}_2-\epsilon$ \footnote{Note that this value of $a$ is positive since $c_H<\theta_2\underline{Q}$.} to reach the perfect path diversity ($\hat{x}_H=0.5$) and  approach optimum social welfare
$$SW_a\rightarrow\theta_0\overline{Q}-\frac{\theta_1\overline{Q}+\theta_2\underline{Q}}{2\theta_2\underline{Q}}c_H.$$
Otherwise we choose $a=1$ to reach optimum social welfare $SW_a=\theta_0\underline{Q}$.

When $\theta_2/\theta_1\le\overline{Q}/\underline{Q}$, in the corresponding medium content-restriction regime we have shown that the maximal social welfare attained at a stable path-diversity equilibrium ($0<\hat{x}_H<1$) is
$$\theta_0\overline{Q}-c_H.$$
In the corresponding strong and weak content-restriction regime, we can infer that the maximal social welfare attained at a stable zero path-diversity equilibrium ($\hat{x}_H=0$ or $1$) is $\theta_0\underline{Q}$. Thus if
$$c_H<\theta_0(\overline{Q}-\underline{Q}),$$
we have
$$\theta_0\overline{Q}-c_H>\theta_0\underline{Q},$$
Then we optimally choose $a=\overline{a}_1$ \footnote{Note that this value of $a$ is positive since $c_H<\theta_1\overline{Q}$.} to reach the perfect path diversity ($\hat{x}_H=0.5$) and the optimum social welfare
$$SW_a=\theta_0\overline{Q}-c_H.$$
Otherwise we choose $a=1$ to reach optimum social welfare $SW_a=\theta_0\underline{Q}$.

Next we prove $PoA_a=1/2$. Using the notation introduced in the proof of Theorem \ref{prop7}, we have
$$PoA_a=\min\limits_{\theta_i,Q(\cdot,1),c_H}\frac{\max\limits_{a}\min\limits_{\hat{x}_H\in NE(a)}\{SW(\hat{x}_H)\}}{SW(x_H^*)}$$
$$\ge\min\limits_{Q(\cdot,1)}\frac{\underline{Q}}{\overline{Q}}\ge\frac{1}{2}$$
since $\min\limits_{\hat{x}_H\in NE(1)}\{SW(\hat{x}_H\}=\theta_0\underline{Q}$ and $SW(x_H^*)\le\theta_0\overline{Q}$.

Now we prove the bound is tight. Assuming maximum user diversity, i.e., $\theta_1$ is near to zero. In this case, $\theta_1Q(x_H,1)-c_H\le a\theta_1Q(x_H,1)$ for any $0\le x_H\le1$ and $0\le a\le1$. Since in any equilibrium, either there is no path diversity or type $\theta_2$ users are indifferent between $H$-path and $L$-path. In the former case, the social welfare is either $\theta_0\underline{Q}$ or $\theta_0\overline{Q}-H$. In the latter case, let $\hat{x}_H$ be any such incentivised equilibrium, then the social welfare is
$$0.5a\theta_1Q(\hat{x}_H,1)+0.5(\theta_2Q(\hat{x}_H,1)-c_H),$$
which is less than
$$\theta_0\overline{Q}-0.5c_H.$$
Thus,
$$\max\limits_{a}\min\limits_{\hat{x}_H\in NE(a)}\{SW(\hat{x}_H)\}\le\max\{\theta_0\overline{Q}-0.5c_H,\theta_0\underline{Q}\}.$$

Consider a specific $Q_1(\cdot)$ function as following
\begin{equation}
\label{fucdef5}
Q_1(x)=\left\{
\begin{array}{ll}
\frac{qx}{\delta} &0\le x\le\delta;\\
q &\delta\le x\le1.
\end{array}
\right.
\end{equation}
where we let $\delta$ be near zero. This content function corresponds to having a finite amount of content, and that we only need one user to collect all the information on a path. Then
\begin{equation}
\label{fucdef6}
Q(x,1)=\left\{
\begin{array}{ll}
\frac{qx}{\delta}+q &0\le x\le\delta;\\
2q &\delta\le x\le1-\delta;\\
q+\frac{(1-x)q}{\delta} &1-\delta\le x\le1.
\end{array}
\right.
\end{equation}
We let $c_H=2\theta_0q$. Then $\theta_0\overline{Q}-0.5c_H=\theta_0q$ and $\theta_0\underline{Q}=\theta_0q$, and it follows that
$$\max\limits_{a}\min\limits_{\hat{x}_H\in NE(a)}\{SW(\hat{x}_H)\}\le\theta_0q.$$
Note that $SW(\hat{x}_H)= SW(\delta)$ for infinitesimal $\delta>0$. Since $SW(\delta)=2\theta_0q-\delta c_H$, we have that the price of anarchy
$$PoA_a\le\frac{\theta_0q}{2\theta_0q-\delta c_H},$$
for infinitesimal $\delta>0$. Hence $PoA_a\le1/2$. This completes the proof.
\end{proof}

A corollary is that diversity in the user types increases the optimum system efficiency.
This is because type-2 users have a larger $\theta_2$ than the average $\theta_0$ and are more sensitive to the restriction on the content than the average user. Hence less restriction is needed to make users switch and obtain  path diversity than in the case of homogeneous users.
And this diversity is necessary to make this equilibrium unique. Of course user diversity does not play a role to improve efficiency if the cost $c_H$ is too large. In Section~\ref{simulation}, we will choose $Q_1(\cdot)$ as (\ref{Q}) and use extensive simulations to show that this mechanism is still efficient in most of the time.

\section{Combined Side-payment and Content-restriction for Heterogeneous Users}\label{combined}

In this section, we propose a combined mechanism: in addition to the payment function $g(x_H)$ in Section~3, we can also use content-restriction with coefficient $a\le1$ in Section~4 as follows.

 \emph{Combined Mechanism:} We collect from each $L$-path participant a payment $g(x_H)$ for only a fraction $a$ of the total information, i.e., $aQ(x_H,b)$, and give each $H$-path participant a subsidy $(b-x_H)g(x_H)/x_H$ to keep the budget balanced.

 This combined mechanism jointly optimizes the payment function $g(x_H)$ and restriction coefficient $a$. Depending on her path choice ($H$ or $L$), a type-$i$ user's payoff  changes from (\ref{payoff}) to
\[
\label{payoff4}
u_i(P,x_H)=\left\{
\begin{array}{ll}
a\theta_iQ(x_H,b)-g(x_H), &\mbox{if }P=L;\\
\theta_iQ(x_H,b)-c_H+\frac{b-x_H}{x_H}g(x_H), &\mbox{if }P=H.
\end{array}
\right.
\]

If we let $a=1$ or $g(x_H)=0$ for any $x_H$, the combined mechanism is simplified to either side-payment or content-restriction. This combined mechanism exploits the \emph{synergies} of the two mechanisms. The use of side-payments reduces the effective cost difference of the two paths. This allows for less content destruction to be necessary for type-2 users to choose the $H$-path (i.e., we can obtain the desired path diversity  in a more efficient way). Of course, since in this equilibrium type-1 users choose the $L$-path and must also pay the incentive fee, we have similar IR issues as we encountered before (see Section~3), but made more acute because the content is of lesser quality for users choosing the $L$-path. We need to consider if it is worthwhile to keep type-1 users in the system, or use a higher side payment and nicely control only type-2 users' choices with some $a<1$.
In the following, we will show how the combined mechanism works. Actually, as long as $c_H$ is not too large, we can show that it is optimal to have a positive flow on every path. More specifically, we have three equilibrium cases: (i) \emph{Case.IR21}: this case ensures both user types' IR for full participation and only type-2 users are indifferent in path choosing, (ii) \emph{Case.IR12}: this case also ensures full participation and only type-1 users are indifferent in choosing between the two paths, and (iii) \emph{Case.IR2}: this case excludes type-1 users by only satisfying type-2's IR due to side-payment.

 \emph{Analysis of Case.IR21}: this case ensures both user types' IR for full participation and only type-2 users are indifferent in path choosing. In this case, if we want to reach an equilibrium $\hat{x}_H$, the equilibrium payoffs of a type-2 by choosing the two paths are equal, i.e.,
$$\theta_2Q(\hat{x}_H,1)-c_H+\frac{1-\hat{x}_H}{\hat{x}_H}g(\hat{x}_H)=a\theta_2Q(\hat{x}_H,1)-g(\hat{x}_H),$$
which is equivalent to
\begin{equation}
\label{condition8}
g(\hat{x}_H)=(c_H-(1-a)\theta_2Q(\hat{x}_H,1))\hat{x}_H.
\end{equation}
Since $\theta_1<\theta_2$, we have
$$g(\hat{x}_H)<(c_H-(1-a)\theta_1Q(\hat{x}_H,1))\hat{x}_H,$$
for type-1 and all users of that type will choose $L$-path. Thus,  we have $x_H^*\in[0,0.5]$. Given $g(\hat{x}_H)$ in (\ref{condition8}), we still need to ensure type-1's IR, i.e.,
\begin{equation}
\label{condition2}
\theta_1aQ(\hat{x}_H,1)-(c_H-(1-a)\theta_2Q(\hat{x}_H,1))\hat{x}_H\ge0.
\end{equation}
As long as (\ref{condition2}) holds, we can incentivize any feasible equilibrium $\hat{x}_H\in(0,0.5]$ by using the following side-payment\footnote{In practice, $g(x_H)$ cannot be infinity, and we can simply replace it by a large enough value.}:
\begin{equation}
\footnotesize
\label{g1}
g(x_H)=\left\{
\begin{array}{ll}
\infty, &\text{if\ }x_H\in[0,\hat{x}_H);\\
(c_H-(1-a)\theta_2Q(\hat{x}_H,1))\hat{x}_H, & \text{if\ }x_H=\hat{x}_H;\\
-\infty, &\text{if\ }x_H\in(\hat{x}_H,1],
\end{array}
\right.
\end{equation}
which becomes (\ref{condition8}) when $x_H$ reaches the stable target $\hat{x}_H$. (\ref{g1}) also helps converge from any $x_H$ to the stable $x_H^*$:
\begin{itemize}
\item If $x_H>\hat{x}_H$, the corresponding (\ref{g1}) becomes negative infinity, $L$-path participants are rewarded with a large enough reward and hence more users are motivated to choose $L$-path and $x_H$ decreases.
\item If $x_H<\hat{x}_H$, the corresponding (\ref{g1}) becomes positive infinity, $L$-path participants are required to pay a large enough penalty and hence more users are motivated to choose $H$-path and $x_H$ increases.
\end{itemize}

Given the side-payment function in (\ref{g1}), now we only need to properly choose $a$ and target $\hat{x}_H$ to maximize the social welfare:
\begin{equation}
\small
\label{maximisation1}
\max\limits_{x_H\in(0,0.5]\atop a\in[0,1]}\big(x_H\theta_2+((\frac{1}{2}-x_H)\theta_2+\frac{1}{2}\theta_1)a\big)Q(x_H,1)-x_Hc_H
\end{equation}
subject to the IR constraint in (\ref{condition2}). We denote the optimal solution to (\ref{maximisation1}) as $(\check{x}_H,\check{a})$ and the corresponding social welfare is $SW(\check{x}_H,\check{a})$.

 \emph{Analysis of Case.IR12:} In this case, the equilibrium payoffs of a type-1 user by choosing the two paths are equal, i.e.,
\[
g(\hat{x}_H)=(c_H-(1-a)\theta_1Q(\hat{x}_H,1))\hat{x}_H.
\]
Since $\theta_2>\theta_1$,
$$g(\hat{x}_H)>(c_H-(1-a)\theta_2Q(\hat{x}_H,1))\hat{x}_H,$$
for type-2 and all users of type-2 will choose $H$-path. Thus, we have $\hat{x}_H\in[0.5,1]$.  We can show that the equilibrium of Case.IR21 is always better than that of Case.IR12. The reason is that if we want to incentivise an equilibrium $\hat{x}_H\in[0.5,1]$ in Case.IR12, we can also incentivise another equilibrium $1-\hat{x}_H\in[0,0.5]$ in Case.IR21 to collect the same information value in (\ref{content1}) but save the total travel cost on the $H$-path.
\begin{lemma}
\label{lemma3}
 The equilibrium of Case.IR12 cannot be better than that of Case.IR21.
\end{lemma}
\begin{proof}
Assume that by choosing $a=\acute{a}$ and $\hat{x}_H=\acute{x}_H\in[0.5,1]$ the social optimum in Case.IR12 is $SW(\acute{x}_H,\acute{a})$. Since at the equilibrium $\hat{x}_H=\acute{x}_H$, all type-1 users are indifferent in path choosing, thus,
$$g(\acute{x}_H)=(c_H-(1-\acute{a})\theta_1Q(\acute{x}_H,1))\acute{x}_H.$$
The payoff of a type-1 user at this equilibrium is
$$\theta_1\acute{a}Q(\acute{x}_H,1)-(c_H-(1-\acute{a})\theta_1Q(\acute{x}_H,1))\acute{x}_H\ge0.$$
All type-2 users will choose $H$-path, the payoff of a type-2 user at this equilibrium is
$$\theta_2Q(\acute{x}_H,1)-c_H+(c_H-(1-\acute{a})\theta_1Q(\acute{x}_H,1))(1-\acute{x}_H)\ge0.$$
Now we calculate the social welfare of each user without considering side-payment due to BB condition.
The payoff of a type-1 user who chooses $L$-path is
$$\theta_1\acute{a}Q(\acute{x}_H,1).$$
The payoff of a type-1 user who chooses $H$-path is
$$\theta_1Q(\acute{x}_H,1)-c_H.$$
The payoff of a type-2 user who must chooses $H$-path is
$$\theta_2Q(\acute{x}_H,1)-c_H.$$

Without consider IR constraint, we can always reach a stable equilibrium $\hat{x}_H=1-\acute{x}_H$ by using
$$a^\prime=\max\{\acute{a},1-\frac{c_H}{\theta_2Q(\acute{x}_H,1)}\},$$
and side-payment function given by (\ref{g1}). At the equilibrium $\hat{x}_H=1-\acute{x}_H$, all type-2 users are indifferent in path choosing, thus,
$$g(1-\acute{x}_H)=(c_H-(1-a^\prime)\theta_2Q(\acute{x}_H,1))(1-\acute{x}_H).$$
The payoff of a type-2 user at this equilibrium is
$$\theta_2Q(\acute{x}_H,1)-c_H+(c_H-(1-a^\prime)\theta_2Q(\acute{x}_H,1))\acute{x}_H\ge0.$$
All type-1 users will choose $L$-path, the payoff of a type-1 user at this equilibrium is
$$\theta_1a^\prime Q(\acute{x}_H,1)-(c_H-(1-a^\prime)\theta_2Q(\acute{x}_H,1))(1-\acute{x}_H)\ge0.$$
Then all type-1 users will participate and IR is satisfied. Now we calculate the social welfare of each user without considering side-payment due to BB condition. The payoff of a type-1 user who must choose $L$-path is
$$\theta_1a^\prime Q(\acute{x}_H,1)\ge\max\{\theta_1\acute{a}Q(\acute{x}_H,1),\theta_1Q(\acute{x}_H,1)-c_H\}.$$
The payoff of a type-2 user who chooses $H$-path is
$$\theta_2Q(\acute{x}_H,1)-c_H.$$
The payoff of a type-2 user who chooses $L$-path is
$$\theta_2a^\prime Q(\acute{x}_H,1)\ge\theta_2Q(\acute{x}_H,1)-c_H.$$
Denote the social welfare attained at this equilibrium as $SW(\acute{x}_H,a^\prime)$. Since at the equilibrium $\hat{x}_H=1-\acute{x}_H$, if we do not take the side-payment into account, each user receive a payoff better than the payoff she receive at the equilibrium $\hat{x}_H=\acute{x}_H$. Therefore, $SW(\acute{x}_H,a^\prime)\ge SW(\acute{x}_H,\acute{a})$.
\end{proof}

 \emph{Analysis of Case.IR2}: this case excludes type-1 users by only satisfying type-2's IR due to side-payment. This is equivalent to the half participation case with $b=0.5$ and $a=1$ in Section \ref{sidepayment}.  The optimal solution in this case is $\tilde{x}_H$ which is given by (\ref{half})  and the corresponding social welfare is $SW(\tilde{x}_H)$.

Actually, besides these three cases, there is another equilibrium case where all type-1 users choose L-path and all type-2 users choose H-path. We can show that this case with $\hat{x}_H=0.5$ is dominated by Case.IR21 as follows.
\begin{lemma}
\label{lemma4}
The equilibrium where all type-1 users choose $L$-path and all type-2 users choose $H$-path cannot be better than the equilibrium of Case.IR21.
\end{lemma}
\begin{proof}
Assume that by choosing $a=\grave{a}$ and $\hat{x}_H=0.5$ we can reach a stable equilibrium where all type-1 users choose $L$-path and all type-2 users choose $H$-path. The corresponding social welfare is denoted as $SW(0.5,\grave{a})$. Since at such equilibrium $\hat{x}_H=0.5$, all type-1 users choose $L$-path, thus,
$$g(0.5)\le0.5(c_H-(1-\grave{a})\theta_1\overline{Q}).$$
All type-2 users choose $H$-path, thus,
$$g(0.5)\ge0.5(c_H-(1-\grave{a})\theta_2\overline{Q}).$$
The payoff of a type-1 user at such equilibrium is
$$\theta_1\grave{a}\overline{Q}-g(0.5)\ge0.$$
It follows that
$$\theta_1\grave{a}\overline{Q}-0.5(c_H-(1-\grave{a})\theta_2\overline{Q})\ge0.$$
Thus, we use $a=\grave{a}$ and side-payment function given by (\ref{g1}) with $\hat{x}_H=0.5$ to reach the same equilibrium ($\hat{x}_H=0.5$) in Case.IR21. The corresponding social welfare are the same since we use the same content-restriction coefficient.
\end{proof}

Finally, the combined mechanism is to compare $SW(\check{x}_H,\check{a})$ in Case.IR21 and $SW(\tilde{x}_H)$ in Case.IR2 to decide to incentivize $\check{x}_H$ or $\tilde{x}_H$. Essentially, this is a tradeoff between the half participation with path diversity and full participation with content-destruction.
\begin{theorem}\label{theorem6}
Depending on the relationship between $SW(\check{x}_H,\check{a})$ and $SW(\tilde{x}_H)$, we decide $a$, $g(x_H)$ and which equilibrium $\hat{x}_H$ to incentivize:
\begin{itemize}
  \item Full participation: If $SW(\check{x}_H,\check{a})\ge SW(\tilde{x}_H)$, it is optimal to keep both types' IR and incentivize $\hat{x}_H=\check{x}$ by choosing $a=\check{a}$ and $g(x_H)$ in (\ref{g1}). The corresponding social welfare is  $SW(\check{x}_H,\check{a})$.
  \item Half participation: If $SW(\check{x}_H,\check{a})<SW(\tilde{x}_H)$, it is optimal to satisfy only type-2's IR and incentivize $\hat{x}_H=\tilde{x}$ by choosing $a=1$ and $g(x_H)$ in (\ref{function}) with $b=0.5$. The corresponding social welfare is $SW(\tilde{x}_H)$.
\end{itemize}
\end{theorem}

To understand why combined mechanism can achieve strictly better social welfare than any of the two individual mechanisms, we consider the payment penalty to $L$-path participants given by (\ref{condition8}). It is smaller than (\ref{condition}) without content-restriction. Note that if $c_H$ and $a$ are small enough, (\ref{condition8}) can be negative, i.e., $L$-path participants may be not penalised but even rewarded. Negative payments on $L$-path occurs because degradation of content on $L$-path is more critical than the cost on $H$-path. This allows us to relax the IR condition on type-1 users and increase the range of equilibria where all users participate.
\begin{proposition}
\label{prop11}
Facing heterogeneous users, the combined mechanism performs significantly better than either the side-payment or content-restriction mechanism, by increasing the price of anarchy $PoA_{ag}$ to more than 0.7 (i.e., $PoA_{ag}\ge0.7$).
\end{proposition}
\begin{proof}
 Notice that the combined mechanism is better than the side-payment mechanism, we have $SW_{ag}\ge SW_g$. The routing equilibrium under optimal side-payment mechanism is better than the routing equilibrium without incentive design, i.e., $SW_g\ge\theta_0\underline{Q}$.
Now we will show
\begin{equation}
\label{swag1}
SW_{ag}\ge\theta_0\overline{Q}-\frac{\theta_1+\theta_2}{2\theta_2}c_H.
\end{equation}
 If $c_H>\theta_2(\overline{Q}-\underline{Q})$, we have $$SW_{ag}\ge SW_g\ge\theta_0\underline{Q}>\theta_0\overline{Q}-\frac{\theta_1+\theta_2}{2\theta_2}c_H.$$
 If $c_H\le\theta_2(\overline{Q}-\underline{Q})$, we can incentivize $0.5$ as a unique and stable equilibrium by choosing $$a=1-\frac{c_H}{\theta_2\overline{Q}},$$
 and $g(x_H)$ given by (\ref{g1}) with $\hat{x}_H=0.5$. In this equilibrium, all type $\theta_1$ users choose $L$-path and receive a payoff given by $\theta_1a\overline{Q}$ while all type $\theta_1$ users choose $H$-path and receive a payoff given by $\theta_2a\overline{Q}$. Then the social welfare attained at this equilibrium is given by
$$0.5(\theta_1+\theta_2)a\overline{Q}=\theta_0\overline{Q}-\frac{\theta_1+\theta_2}{2\theta_2}c_H.$$

To show $PoA_{ag}\ge0.7$, we need to show
$$\min\limits_{\theta_i,c_H,Q(\cdot,1)}\frac{SW_{ag}}{SW(x_H^*)}\ge0.7.$$
Now we assume there exists an instance which is defined by the parameters $\theta_1$ $\theta_2$, $c_H$ and function $Q_1(\cdot)$ satisfies that
$$\frac{SW_{ag}}{SW(x_H^*)}=\beta,\ \beta\in[0.5,0.75].$$
Without loss of generality, we assume in this instance $\theta_1=\alpha\in[0,1)$, $\theta_2=1$, $\underline{Q}=1$ and the social optimum is attained for $x_H^*=\delta\in(0,0.5)$. We let $Q_1(\delta)=s$ and $Q_1(1-\delta)=t$ where $s,t$ should satisfy $\delta\le s\le t$, $1-\delta\le t\le 1$ and $\delta t\le(1-\delta)s$ due to our assumption on $Q_1(x)$. Then the social optimum is $$SW(x_H^*)=\frac{(1+\alpha)(s+t)}{2}-\delta c_H.$$
Because of $SW_{ag}\ge \theta_0\underline{Q}$, $SW_{ag}\ge SW_g$ and (\ref{swag1}), one can easily check the following
\begin{enumerate}
  \item The routing equilibrium without incentive design: $SW=(1+\alpha)/2$ should be less than $\beta SW(\delta)$, thus, $$\frac{1+\alpha}{2}<\frac{\beta(1+\alpha)(s+t)}{2}-\beta\delta c_H$$
      It follows that
      $$c_H<\frac{(1+\alpha)(\beta(s+t)-1)}{2\beta\delta}.$$
  \item Content-restriction with $a=1-c_H/(\theta_2\overline{Q})$ and $g(x_H)$ given by (\ref{g1}) with $\hat{x}_H=0.5$: \begin{eqnarray}
      SW_{ag}&\ge&\theta_0\overline{Q}-\frac{\theta_1+\theta_2}{2\theta_2}c_H=\frac{(1+\alpha)(\overline{Q}-c_H)}{2}\nonumber\\
      &\ge&\frac{(1+\alpha)(s+t-c_H)}{2},\nonumber
      \end{eqnarray}
       which should be less than $\beta SW(\delta)$, it follows that
       $$c_H>\frac{(1-\beta)(1+\alpha)(s+t)}{(1+\alpha)-2\beta\delta}.$$
  \item Full participation case with $a=1$ and $b=1$: Let $r=(3s+3t-4)/(4s+4t-4)$, and note that
        $$r(s+t)+(1-r)=\frac{3}{4}(s+t),\ r<\frac{3}{4}.$$
        Since
        \begin{eqnarray}
        &&\theta_0Q(r\delta,1)-r\delta c_H c_H\nonumber\\
        &\ge&\theta_0(rQ(\delta,1)+(1-r)Q(0,1))-\frac{3}{4}\delta\nonumber\\
        &\ge&\beta SW(\delta),\nonumber
        \end{eqnarray}

        IR constraint can not hold for type-1 at $x_H=r\delta$. Thus,
        $$\frac{3}{4}\theta_1(\delta,1)-r\delta c_H\le\theta_1Q(r\delta,1)-r\delta c_H<0.$$
        It follows that
        $$c_H>\frac{\alpha(3s+3t-3)(s+t)}{(3s+st-4)\delta}.$$
  \item Half participation case with $a=1$ and $b=0.5$: If $\delta\in[0,0.25]$, we can obtain at least
        $$0.5\theta_2Q(\delta,0.5)-\delta c_H$$
        which should be less than $\beta SW(\delta)$. Note that $0.5-\delta\ge\delta$ and hence $$Q_1(0.5-\delta)\ge\frac{1}{2-4\delta}Q_1(\delta)+\frac{1-4\delta}{2-4\delta}Q_1(1-\delta).$$
        It follows that
        $$c_H>\frac{(3-4\delta)s+(1-4\delta)t-\beta(1+\alpha)(s+t)(2-4\delta)}{2\delta(1-\beta)(2-4\delta)},$$
        when $\delta\in[0,0.25]$. If $\delta\in[0.25,0.5]$, we can obtain at least
        $$0.5\theta_2Q(0.25,0.5)-0.25 c_H\ge\frac{\theta_2Q_1(\delta)}{4\delta}-0.25c_H,$$
        which should be less than $\beta SW(\delta)$. It follows that either
        $$\beta\ge\frac{1}{4\delta},$$
        or
        $$c_H>\frac{s-2\delta\beta(1+\alpha)(s+t)}{\delta-4\beta\delta^2}.$$
\end{enumerate}

When $\delta\in(0,0.25]$, we have
\begin{equation}
\label{fucdef5}
\left\{
\begin{array}{l}
c_H<\frac{(1+\alpha)(\beta(s+t)-1)}{2\beta\delta},\\
c_H>\frac{(1-\beta)(1+\alpha)(s+t)}{(1+\alpha)-2\beta\delta},\\
c_H>\frac{\alpha(3s+3t-3)(s+t)}{(3s+st-4)\delta},\\
c_H>\frac{(3-4\delta)s+(1-4\delta)t-\beta(1+\alpha)(s+t)(2-4\delta)}{2\delta(1-\beta)(2-4\delta)}.
\nonumber\end{array}
\right.
\end{equation}
It follows that, 
 {
  \footnotesize
 \begin{equation}
 \beta>\max\{\frac{-3s-3t+4\delta-3s\delta-3t\delta}{(s+t)(-3s-3t-2\delta +3s\delta+3t\delta)},\frac{24 s - 21 s^2 + 16 t - 36 s t - 15 t^2 - 40 s \delta +
 36 s^2 \delta - 40 t \delta + 72 s t \delta + 36 t^2 \delta}{(s+t)(12 - 6 s - 3 s^2 - 2 t - 12 s t - 9 t^2 - 24 \delta +
 8 s \delta + 12 s^2 \delta + 8 t \delta + 24 s t \delta +
 12 t^2 \delta)}\}.\nonumber
\end{equation}
 }By minimizing the right hand side of 
 {
  \footnotesize
 \begin{equation}
\max\{\frac{-3s-3t+4\delta-3s\delta-3t\delta}{(s+t)(-3s-3t-2\delta +3s\delta+3t\delta)},\frac{24 s - 21 s^2 + 16 t - 36 s t - 15 t^2 - 40 s \delta +
 36 s^2 \delta - 40 t \delta + 72 s t \delta + 36 t^2 \delta}{(s+t)(12 - 6 s - 3 s^2 - 2 t - 12 s t - 9 t^2 - 24 \delta +
 8 s \delta + 12 s^2 \delta + 8 t \delta + 24 s t \delta +
 12 t^2 \delta)}\}\nonumber
\end{equation}
 }in the feasible region of $s,t,\delta$, i.e.,
\begin{equation}
\delta\le s\le t,\ 1-\delta\le t\le 1,\ \delta t\le(1-\delta)s,\ \delta\in(0,0.25],\nonumber
\end{equation}
we get $\beta>0.7$.

When $\delta\in[0.25,0.5)$, we have either
\begin{equation}
\label{fucdef6}
\left\{
\begin{array}{l}
c_H<\frac{(1+\alpha)(\beta(s+t)-1)}{2\beta\delta},\\
c_H>\frac{(1-\beta)(1+\alpha)(s+t)}{(1+\alpha)-2\beta\delta},\\
c_H>\frac{\alpha(3s+3t-3)(s+t)}{(3s+st-4)\delta},\\
c_H>\frac{s-2\delta\beta(1+\alpha)(s+t)}{\delta-4\beta\delta^2},
\nonumber\end{array}
\right.
\end{equation}
or
\begin{equation}
\label{fucdef7}
\left\{
\begin{array}{l}
c_H<\frac{(1+\alpha)(\beta(s+t)-1)}{2\beta\delta},\\
c_H>\frac{(1-\beta)(1+\alpha)(s+t)}{(1+\alpha)-2\beta\delta},\\
c_H>\frac{\alpha(3s+3t-3)(s+t)}{(3s+st-4)\delta},\\
\beta\ge\frac{1}{4\delta}.
\nonumber\end{array}
\right.
\end{equation}
It follows that either
{
  \footnotesize
 \begin{equation}
 \beta>\max\{\frac{-3s-3t+4\delta-3s\delta-3t\delta}{(s+t)(-3s-3t-2\delta +3s\delta+3t\delta)},\frac{1}{4\delta}\}.\nonumber
\end{equation}
 }or 
{
  \footnotesize
 \begin{equation}
 \beta>\max\{\frac{-3s-3t+4\delta-3s\delta-3t\delta}{(s+t)(-3s-3t-2\delta +3s\delta+3t\delta)},\frac{-7 s + 6 s^2 - 3 t + 9 s t +
 3 t^2}{(s + t) (-s - 3 t + 3 s t + 3 t^2 - 12 \delta +
   12 s \delta + 12 t \delta)}\}.\nonumber
\end{equation}
 }By minimizing
 {
  \footnotesize
 \begin{equation}
\max\{\frac{-3s-3t+4\delta-3s\delta-3t\delta}{(s+t)(-3s-3t-2\delta +3s\delta+3t\delta)},\frac{1}{4\delta}\}\nonumber
\end{equation}
 }and
{
  \footnotesize
 \begin{equation}
 \max\{\frac{-3s-3t+4\delta-3s\delta-3t\delta}{(s+t)(-3s-3t-2\delta +3s\delta+3t\delta)},\frac{-7 s + 6 s^2 - 3 t + 9 s t +
 3 t^2}{(s + t) (-s - 3 t + 3 s t + 3 t^2 - 12 \delta +
   12 s \delta + 12 t \delta)}\}\nonumber
\end{equation}
 }in the feasible region of $s,t,\delta$, i.e.,
  \begin{equation}
\delta\le s\le t,\ 1-\delta\le t\le 1,\ \delta t\le(1-\delta)s,\ \delta\in[0.25,0.5),\nonumber
\end{equation}
we get $\beta>0.7$.

In conclusion, we have proven that if $\beta\le0.75$, it must follow that $\beta>0.7$. Therefore, $PoA_{ag}>0.7$.
\end{proof}
\vspace{-5pt}
\section{Simulation Results}\label{simulation}
\begin{figure}[!t]
  \centering
  \includegraphics[width=3in]{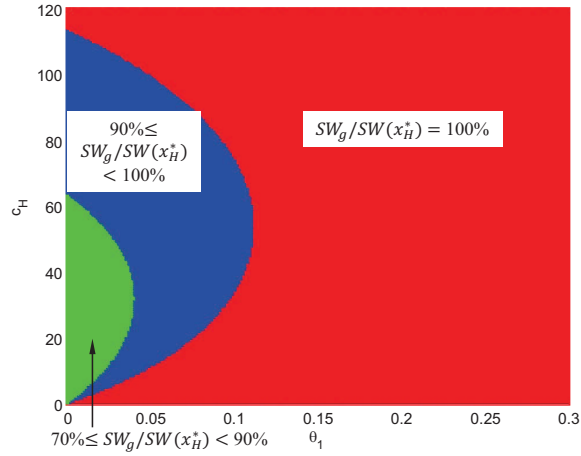}
  \caption{The ratio $SW_g/SW(x_H^*)$ between the optimum SW value under the side-payment and the social optimum $SW(x_H^*)$. This ratio is larger than 70\% globally. We set $\theta_0=0.5$ and use $Q_1(x_H)$ from (\ref{Q}). We observe  inefficiency when $\theta_1$ decreases while $c_H$ remain at some moderately low value}
  \label{swg}
\end{figure}
\begin{figure}[!t]
 \centering
  \includegraphics[width=3in]{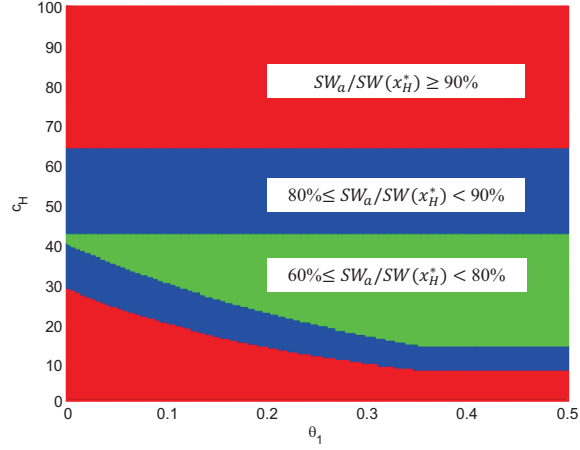}
  \caption{The ratio $SW_a/SW(x_H^*)$ between the optimum SW value under the content-restriction and the social optimum $SW(x_H^*)$. The ratio is larger than 60\% globally. We set $\theta_0=0.5$ and use $Q_1(x_H)$ from (\ref{Q}). .}
  \label{swa}
\end{figure}
In this section, we use extensive simulations to show the performances of the side-payment and content-restriction mechanisms in the heterogeneous user case.\footnote{We skip the homogeneous user case here, as Proposition~\ref{theorem1} shows that the side-payment already achieves the social optimum there.} Recall that both these two mechanisms can be as bad as the original system without incentives, as we have $PoA_g=PoA_a=1/2$ in the worst case by choosing an arbitrary content function $Q_1(x_H)$. In the following, we reasonably specify $Q_1(x_H)$ as (\ref{Q}) and examine the performances by varying the other parameter values.

 Figure \ref{swg} compares the maximum social welfare achieved by the side-payment with the social optimum, by showing the ratio $SW_g/SW(x_H^*)$ under different $\theta_1$ and $c_H$ values. By fixing the average valuation $\theta_0$ as 0.5, we can vary $\theta_1$ value to change user diversity. We observe that as long as $\theta_1$ is large, our side-payment mechanism can achieve the social optimum with $SW_g/SW(x_H^*))=100\%$. This explains this mechanism performs well when users are not diverse. Even when $\theta_1$ is small, this mechanism is still efficient to achieve $SW_g/SW(x_H^*))\geq70\%$. When $c_H$ is zero, users will automatically reach path-diversity and this ratio is 1 in Fig. \ref{swg}. When $c_H$ is too large, both this mechanism and the social optimum do not assign any user to explore the H-path. Thus, the ratio is also 1.

 Figure \ref{swa} compares the maximum social welfare achieved by the content-restriction with the social optimum, by showing the ratio $SW_a/SW(x_H^*))$ under different $\theta_1$ and
$c_H$ values. Unlike the side-payment, this mechanism performs well when $\theta_1$ is small and users are diverse. Even when $\theta_1$ is large, this mechanism is still efficient to achieve $SW_a/SW(x_H^*))\geq 60\%$. Similar to Fig.~\ref{swg}, when $c_H$ is zero or very high, the ratio is 1 in Fig. \ref{swa} without any efficiency loss.
\begin{figure}[!t]
  \centering
  \includegraphics[width=3in]{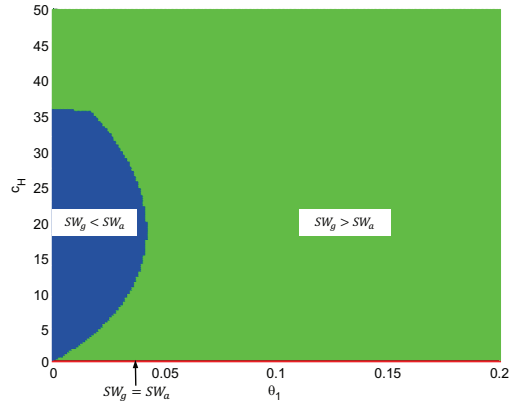}
  \caption{Performance comparison between side-payment and content-restriction in terms of the achieved social welfare under different $\theta_1$ and $c_H$ values. We observe that in most cases we have $SW_g> SW_a$,
and $SW_g< SW_a$ only when users are diverse (with small $¦È_1$) and $c_H$ is small. Here, we set $\theta_0=0.5$ and use specified $Q_1(x_H)$ in (\ref{Q}).}
  \label{figure5}
\end{figure}

 As $PoA_a=PoA_g=1/2$ in the worst case, it is difficult to analytically compare the two mechanisms' performances. Figure \ref{figure5} numerically compares the side-payment and content-restriction mechanisms in terms of social welfare ($SW_a$ versus $SW_g$) under different values of $\theta_1$ and $c_H$. We observe that in most cases we have $SW_g>SW_a$, and $SW_g<SW_a$ only when users are diverse (with small $\theta_1$) and $c_H$ is small. When users are diverse, the type-1 users' IR becomes a problem under the side-payment, which may sacrifice type-1 users' participation to reach path-diversity in the half-participation case. Differently, the content-restriction mechanism does not worry about IR and can achieve prefect path-diversity in the full-participation case. Thus, $SW_g<SW_a$ when both $\theta_1$ and $c_H$ are small. Note that when $c_H$ is small, to achieve the path-diversity, the content-restriction does not need to destruct much information to motivate users to H-path.
\begin{figure}[!t]
  \centering
  \includegraphics[width=3in]{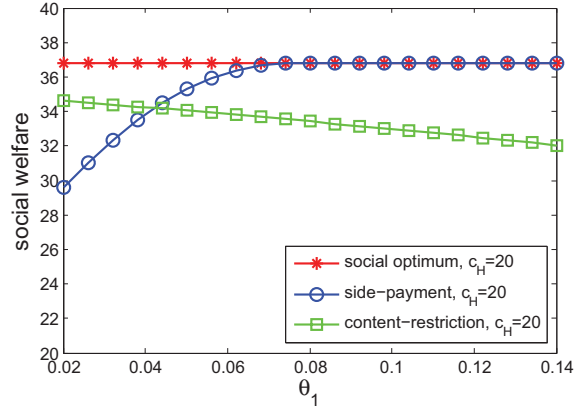}
  \caption{Comparison between $SW(x_H^*)$, $SW_g$ and $SW_a$ under different values of $\theta_1$. We observe $SW_a$ $>SW_g$ for small $\theta_1$ and $SW_a<SW_g$ for large $\theta_1$. We set $\theta_0=0.5$ and use $Q_1(x_H)$ in (\ref{Q}).}
  \label{figure6}
\end{figure}

Figure \ref{figure6} further shows the social welfare values achieved by the two mechanisms and compares to the social optimum. When $c_H$ is small (e.g., $c_H=20$ in Fig. \ref{figure6}), we have $SW_a>SW_g$ for small $\theta_1\in [0, 0.042]$. But when users are less diverse and $\theta_1$ increases beyond 0.042, we have $SW_a<SW_g$. When $\theta_1$ is large enough, $SW_g$ achieves the social optimum $SW(x_H^*))$.
\section{Network Model Generalization}\label{generalization}
Unlike our simple two-path model in Fig.~1, people in practice may have more than two choices of routes if they would like to travel from one point to another. In addition, different routes may have some overlap with each other. In this section, we consider a more general network with $K\geq 3$ parallel paths following a path overlap in Fig.~8, and want to answer the following questions:
\begin{itemize}
\item How does the path multiplicity change the original social efficiency without incentive design?
\item How to apply the side-payment and content-restriction mechanisms to this generalized network model and what are the new challenges?
\item Will our main results developed in previous sections still hold here and what are the new insights?
\end{itemize}
\subsection{New Network Model and Equilibrium}
\begin{figure}[!t]
  \centering
  \includegraphics[width=3in]{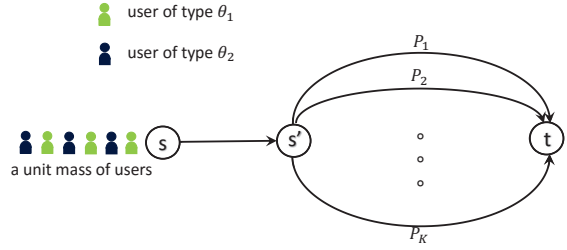}
  \caption{ There are $K$ possible routes from $s$ to $t$ and all routes contain a common link from $s$ to intermediary $s^\prime$ which contains $N_0$ pieces of information.  From $s^\prime$ to $t$, users must select one among $K$ possible links,  denoted as $P_1$, $P_2$, $\cdots$, $P_K$. Each of these parallel links uniformly contains $N/K$ pieces of information, keeping the total information constant. }
  \label{figuremulti}
\end{figure}
As shown in Fig.~\ref{figuremulti}, we consider $K$ routes from $s$ to $t$ and all routes have a common link from $s$ to intermediary $s^\prime$ which contains $N_0$ pieces of information. As all users go through this common link, their collected information is
\begin{equation}
\label{Q0}
Q_0=N_0\big(1-(1-\frac{\phi}{N_0})^{n}\big),
\nonumber\end{equation}
which is calculated as in (\ref{Q}). Without much loss of generality, we assume traveling on the common link from $s$ to $s^\prime$ generates no cost. From $s^\prime$ to $t$, users must select one out of $K$ links, denoted as $P_1$, $P_2$, $\cdots$, $P_K$. Each of these parallel links uniformly contains $N/K$ pieces of information. To keep the average travel cost over all $K$ paths at $c_H/2$, i.e., the same as in the prior two-path model,, we set the cost over any path $P_k$ as
$$c_k=\frac{k-1}{K-1}c_H, \ \ \ \forall k\in\{1,2,\cdots,K\}.$$
We summarize the unit mass of users' routing choices as non-negative flows $(x_1,x_2,\cdots,x_K)$ such that $\sum_{k=1}^K x_k=1$.
Similar to (\ref{content1}) in the two-path model, we define the total information content collected and aggregated by the users from the $K$ routes as $Q(x_1,x_2,\cdots,x_K)$, which is given by
\begin{equation}
Q(x_1,x_2,\cdots,x_K)=Q_0+\sum\limits_{k=1}^{K}Q_1(x_k),\nonumber
\end{equation}
where $Q_1(x_k)$ denotes the 
information content collected over a single link $P_k$ when a fraction $x_i$ of the users travel over that link. Similar to the specific $Q_1(\cdot)$ function in (\ref{Q}) defined in Section \ref{systemmodel}, now we have
\begin{equation}
Q_1(x_k)=\frac{N}{K}\big(1-(1-\frac{K\phi}{N})^{nx_k}\big) \,.\nonumber
\end{equation}
We assume all users are of the same type (i.e., $\theta_1=\theta_2=\theta_0$) for tractability of analysis.\footnote{We plan to study the two heterogeneous user types for the general network model in the future, where the IR constraint may be violated for type-1 users as in side-payment mechanism design in Section~\ref{subsection:heterogeneous}.}
Similar to (\ref{payoff}) in Section \ref{systemmodel}, here the payoff  function of a homogeneous user by choosing path $P_k$ is
\begin{equation}
u(P_k,x_1,x_2,\cdots,x_K)=\theta_0Q(x_1,\cdots,x_K)-c_k.
\end{equation}
In the following, we similarly define the routing equilibrium $(\hat{x}_1,\hat{x}_2,\cdots,\hat{x}_K)$ as in Definition \ref{def1}.
\begin{definition}
\label{def2}
A feasible flow $(\hat{x}_1,\hat{x}_2,\cdots,\hat{x}_K)$ in the content routing game  is an equilibrium if no user traveling over any of the $K$ paths will
profit by  deviating from her current path choice to increase her payoff.
\end{definition}

The social welfare now becomes
\begin{equation}
SW(x_1,x_2\cdots,x_K)=\theta_0Q(x_1,x_2\cdots,x_K)-\sum\limits_{k=1}^{K}c_kx_k,\label{eq:SWnew}
\end{equation}
and the social optimum flow $(x_1^*,x_2^*\cdots,x_K^*)$ is to maximize (\ref{eq:SWnew}). 

Since path $P_1$ does not incur any cost ($c_1=0$), all users are better off by choosing path $P_1$. We have the following result by generalizing Propositions~\ref{prop} and \ref{prop2} (where $K=2$) via a similar proof analysis.
\begin{proposition}
There exists a unique equilibrium for the content routing game: $\hat{x}_1=1$ and $\hat{x}_2=\cdots=\hat{x}_K=0$. The price of anarchy of the content routing game without incentive design is $PoA=1/K$.
\end{proposition}
\begin{proof}
Let $S_K$ denote the set of all feasible flow partitions, i.e.,
$$S_K:=\{(x_1,x_2,\cdots,x_K)|x_1,x_2,\cdots,x_K\in[0,1]\mbox{ and } x_1+x_2+\cdots+x_K=1\}.$$
Note that the maximal quantity of content is attained by letting $1/K$ flow go through each path, hence
\begin{eqnarray}
SW(x_1^*,x_2^*,\cdots,x_K^*)&\le&\max\limits_{(x_1,x_2,\cdots,x_K)\in S_K}\{\theta_0Q(x_1,x_2,\cdots,x_K)\}\nonumber\\
&=&3\theta_0Q_1(\frac{1}{K})+\theta_0Q_0.\nonumber
\end{eqnarray}
Thus the price of anarchy is lower-bounded as follows:
$$PoA\ge\min\limits_{\theta_i,Q_1(\cdot),c_H}\frac{\theta_0Q_1(1)+\theta_0Q_0}{K\theta_0Q_1(\frac{1}{K})+\theta_0Q_0}\ge\frac{1}{K}.$$

We can show that the bound is also tight and hence $PoA=1/K$  by constructing a suitable function $Q_1(x)$. Let
\begin{equation}
\label{fucdef}
Q_1(x)=\left\{
\begin{array}{ll}
Kqx &0\le x\le\frac{1}{K};\\
q &\frac{1}{K}\le x\le1.
\end{array}
\right.
\end{equation}
Consider $Q_0=0$ and any fixed $c_H>0$. Using this setup, the $PoA$ for this specific instance becomes
$$PoA\le\min\limits_{\theta_i,c_H}\frac{\theta_0Q_1(1)}{SW(\frac{1}{K},\frac{1}{K},\cdots,\frac{1}{K})}=\min\limits_{\theta_i,c_H}\frac{\theta_0q}{K\theta_0q-\sum\limits_{k=1}^{K}\frac{c_k}{K}}=\frac{1}{K}.$$
This completes the proof that $PoA=1/K$.
\end{proof}

The equilibrium completely ignores the content on all other paths except path $P_1$, leading to a low $POA=1/K$.  As the path multiplicity $K$ increases to infinity, the price of anarchy and social efficiency decrease to zero. Thus, mechanism design is more needed in a multi-path network.  The path overlap does not make any difference because all users must travel through the overlapping path and collect the fixed content $Q_0$. For illustration purpose and without much loss of generality, in the following mechanism design we focus on $K=3$ and also explain the results for arbitrary $K$.
\subsection{Side-payment as Incentive}
When applying side-payment to the multi-path model, we design individual side-payment functions $g^k(x_1,x_2,\cdots,x_K)$ for participants over path $P_k$ with $k\in\{1,2,\cdots,K\}$. 
Unlike payment function \eqref{function} in Proposition~\ref{lemma1}, now it is more complicated to design $K$ many side-payment functions to incentivize the social welfare maximizer $(x_1^*,x_2^*\cdots,x_K^*)$ as the equilibrium target, while maintaining the budget balanced. Moreover, given a user has more routing options to deviate from the existing choice, the stability of equilibrium is more difficult to prove even in the three-path network. In that respect we have introduced more sophisticated Lyapunov function techniques to ensure and prove the target equilibrium's asymptotic stability. This allows us to successfully design side-payment functions for the three paths which are shown in Proposition \ref{prop6}. For the more general case $K>3$, similar approach applies.  

\begin{proposition}
\label{prop6}
In a three-path network model with travel costs $(c_1, c_2, c_3)=(0, c_H/2, c_H)$, to incentivize users to choose paths $P_2$ and $P_3$  and converge from any initial flow to the social welfare maximizer $(x_1^*, x_2^*, x_3^*)$,  we design side-payment functions on paths $P_1$ and $P_2$  as
\begin{equation}
\label{payment2}
\left\{
\begin{array}{l}
g_{(x_1^*,x_2^*,x_3^*)}^1(x_1,x_2,x_3)=\frac{(1-x_2-x_3)(x_2^*c_2+x_3^*c_3)}{1-x_2^*-x_3^*}, \\
g_{(x_1^*,x_2^*,x_3^*)}^2(x_1,x_2,x_3)=-\frac{(1-x_2-x_3)^2x_2^*c_2}{(1-x_2^*-x_3^*)x_2}+x_3(c_3-c_2).
\end{array}
\right.
\end{equation}
Note that as expected the side-payment on path $P_1$ is larger than that on $P_2$ (i.e., $g^1>g^2$). This side-payment mechanism perfectly achieves price of anarchy $PoA_g$ $=1$ without any loss of system efficiency.
\end{proposition}
\begin{proof}
To enforce BB, the money we refund to each path participant over path $P_3$ is,
\begin{equation}\label{Hpayment}
\frac{(1-x_2-x_3)^2x_3^*c_3}{(1-x_3^*-x_2^*)x_3}-x_2(c_2-c_2).
\end{equation}
Given \eqref{Hpayment} and the side-payment function (\ref{payment2}), a user's payoff now becomes
\begin{equation}
\label{payoffside}
u(P_k,x_1,x_2,x_3)=\left\{
\begin{array}{ll}
\theta_0Q(x_1,x_2,x_3)-\frac{x_1(x_2^*c_2+x_3^*c_3)}{x_1^*}, &\mbox{if }k=1;\\
\theta_0Q(x_1,x_2,x_3)+\frac{x_1^2x_2^*c_2}{x_1^*x_2}-x_3(c_3-c_2)-c_2, &\mbox{if }k=2;\\
\theta_0Q(x_1,x_2,x_3)+\frac{x_1^2x_3^*c_3}{x_1^*x_3}-x_2(c_2-c_3)-c_3, &\mbox{if }k=3.
\end{array}
\right.
\end{equation}
where $x_1=1-x_2-x_3$, $x_1^*=1-x_2^*-x_3^*$.
One can check that when $(x_1,x_2,x_3)=(x_1^*,x_2^*,x_3^*)$, each user's payoff is
$$u(P_k,x_1^*,x_2^*,x_3^*)=\theta_0Q(x_1^*,x_2^*,x_3^*)-x_2^*c_2-x_3^*c_3\ge0,$$
for any $P_k\in\{P_1,P_2,P_3\}$. Thus, every user receives the same payoff no matter which path she chooses and hence $(x_1^*,x_2^*,x_3^*)$ is an equilibrium. Note that each user's equilibrium payoff is equal to the optimal social welfare which is nonnegative, IR is satisfied. IC is also satisfied since the payment function only depends on the path choice and is irrelevant to users' types.

Next we prove  $(x_1^*,x_2^*,x_3^*)$ is a unique equilibrium. One can check that if $(\hat{x}_1,\hat{x}_2,\hat{x}_3)$ is an equilibrium, then $\hat{x}_k\not=0$ for any $k\in\{1,2,3\}$. For example, we assume $\hat{x}_1=0$, then the payoff function \eqref{payoffside} evaluated at $(\hat{x}_1,\hat{x}_2,\hat{x}_3)$ is
\begin{equation}
u(P_k,\hat{x}_1,\hat{x}_2,\hat{x}_3)=\left\{
\begin{array}{ll}
\theta_0Q(\hat{x}_1,\hat{x}_2,\hat{x}_3), &\mbox{if }k=1;\\
\theta_0Q(\hat{x}_1,\hat{x}_2,\hat{x}_3)-\hat{x}_3(c_3-c_2)-c_2, &\mbox{if }k=2;\\
\theta_0Q(\hat{x}_1,\hat{x}_2,\hat{x}_3)-\hat{x}_2(c_2-c_3)-c_3, &\mbox{if }k=3.
\end{array}
\right.
\nonumber\end{equation}
 Note that now $u(P_3,\hat{x}_1,\hat{x}_2,\hat{x}_3)=u(P_2,\hat{x}_1,\hat{x}_2,\hat{x}_3)<u(P_1,\hat{x}_1,\hat{x}_2,\hat{x}_3)$, thus all users have incentive to switch to $L$-path. Hence, $\hat{x}_1\not=0$. Similarly, we can prove that $\hat{x}_2\not=0$ and $\hat{x}_3\not=0$. Thus, at an equilibrium $(\hat{x}_1,\hat{x}_2,\hat{x}_3)$, each path has positive flow. As a result, we have
 \begin{equation}\label{eq:side}
u(P_3,\hat{x}_1,\hat{x}_2,\hat{x}_3)=u(P_2,\hat{x}_1,\hat{x}_2,\hat{x}_3)=u(P_1,\hat{x}_1,\hat{x}_2,\hat{x}_3).
 \end{equation}
 The only feasible solution to \eqref{eq:side} is
 $$(\hat{x}_1,\hat{x}_2,\hat{x}_3)=(x_1^*,x_2^*,x_3^*).$$
 Therefore, we have proven that $(x_1^*,x_2^*,x_3^*)$ is a unique equilibrium.

 Finally, we will prove $(x_1^*,x_2^*,x_3^*)$ is a stable equilibrium. In \cite{smith1984stability}, the author studied the stability of Wardrop's equilibrium. Firstly, the author defined the dynamics of the flow: the flows switch form path $P$ to $P'$ at rate $x_P\max\{0,C_P-C_{P'}\}$ where $C_P$ is the cost on $P$-path. Then Lyapunov function is used to obtain the stability results. Now we will modify the payoff defined in \eqref{payoffside} in order to apply the results in \cite{smith1984stability}. Note that in \eqref{payoffside}, the content part $\theta_0Q(x_1,x_2,x_3)$ is the same for each user no matter which path she chooses. Thus, our network model is equivalent to having cost functions defined on each path as follows,
\begin{equation}
C_P(x_1,x_2,x_3)=\left\{
\begin{array}{ll}
\frac{x_1(x_2^*c_2+x_3^*c_3)}{x_1^*}, &\mbox{if }P=P_1;\\
-\frac{x_1^2x_2^*c_2}{x_1^*x_2}+x_3(c_3-c_2)+c_2, &\mbox{if }P=P_2;\\
-\frac{x_1^2x_3^*c_3}{x_1^*x_3}+x_2(c_2-c_3)+c_3, &\mbox{if }P=P_3.
\end{array}
\right.
\nonumber\end{equation}
Define $\bm{C}(x_1,x_2,x_3):=(C_{P_1}(x_1,x_2,x_3),C_{P_2}(x_1,x_2,x_3),C_{P_3}(x_1,x_2,x_3))$. Now we can use exactly the same dynamical system (5) defined in \cite{smith1984stability}. To prove the stability of this dynamical system,  we only need to prove $\bm{C}(\cdot)$ is monotone.
Let $\bm{J}$ be the Jacobian matrix of $\bm{C}(\cdot)$ evaluated at $(x_1,x_2,x_3)$, then
\begin{gather*}
\bm{J}=\begin{bmatrix}
\frac{x_2^*c_2+x_3^*c_3}{x_1^*}&0&0\\
-\frac{2x_1x_2^*c_2}{x_2x_1^*}&\frac{x_1^2x_2^* c_2}{x_3^2x_1^*}&-c_2+c_3\\
-\frac{2x_1x_3^*c_3}{x_3x_1^*}&c_2-c_3&\frac{x_1^2x_3^* c_3}{x_3^2x_1^*}
\end{bmatrix}.
\end{gather*}
One can easily check that $\bm{J}$ is positive definite. Thus,  $(x_1^*,x_2^*,x_3^*)$ is asymptotically stable as an equilibrium.

 Therefore, we can obtain the social optimal flow $(x_1^*,x_2^*,x_3^*)$  as a stable and unique equilibrium by applying the side-payment function defined in (\ref{payment2}). In addition, IR, IC, and BB are all satisfied.
\end{proof}

With the help of the above side-payment mechanism, we greatly increase the $POA$ from 1/3 (Proposition~5) to 1. 
The path multiplicity does not change the side-payment's efficiency and $PoA=1$ holds for arbitrary $K$.

\subsection{Content-restriction as Incentive}
We now apply content-restriction to the multi-path network model. To motivate more users to choose the higher cost paths ($P_2,\cdots,P_K$) instead of zero cost $P_1$, the system planner should only provide a fraction $a_k\in[0,1]$ of the total information content to participants over the lower-cost path $P_k\in\{P_1, \cdots, P_{K-1}\}$.
As lower cost path participants need more incentives to change routes, we expect that $a_k$ increases with the path index $k$.

Unlike the two-path model in Section~\ref{contentrestriction}, more paths here imply more possible equilibria when under the content-restriction mechanism.
In addition, multiple paths may challenge the stability of the equilibria. We may obtain infinitely many equilibria but only a subset of them is stable. Still, we design Lyapunov function to ensure and prove that our content-restriction coefficients $a_k$'s lead the user flow partition to the stable subset of equilibria (instead of a particular one), where each equilibrium in this stable set \emph{attains the same total system efficiency}. In the following proposition, we design content-restriction coefficients for three-path case. For a more general case ($K>3$), one can derive the content-restriction coefficients analogously although the analysis will be more tedious.

\begin{proposition}\label{prop7}
The content-restriction operates differently according to the travel cost distribution:
\begin{itemize}
  \item In low cost regime (i.e., $c_H\leq 2\theta_0(3Q_1(\frac{1}{3})-2Q_1(\frac{1}{2}))$,
it is optimal to choose 
   \begin{equation}
  a_1=1-\frac{c_3}{3\theta_0Q_1(\frac{1}{3})+\theta_0Q_0-\epsilon},\ a_2=1-\frac{c_3-c_2}{3\theta_0Q_1(\frac{1}{3})+\theta_0Q_0-\epsilon},
   \nonumber
  \end{equation}
  with $a_1<a_2$ to approach the path diversity equilibrium $(\hat{x}_1,\hat{x}_2, \hat{x}_3)=$ $(1/3,1/3,1/3)$ among the three paths, and the optimal social welfare is $SW_a\rightarrow3\theta_0Q_1(1/3)+\theta_0Q_0-c_3$.
  \item In medium cost regime ($2\theta_0(3Q_1(\frac{1}{3})-2Q_1(\frac{1}{2}))< c_H\leq  2\theta_0(2Q_1(\frac{1}{2})-Q_1(1)) $),
it is optimal to choose 
    \begin{equation}
  a_1=1-\frac{c_2}{2\theta_0Q_1(\frac{1}{2})+\theta_0Q_0-\epsilon},\ a_2=1\nonumber
  \end{equation}
to approach only the path diversity between path $P_1$ and path $P_2$ with equilibrium $(\hat{x}_1,\hat{x}_2, \hat{x}_3)=(1/2,1/2,0)$, and optimal social welfare is $SW_a\rightarrow2\theta_0Q_1(1/2)+\theta_0Q_0-c_2$.
\item In high cost regime ($c_H> 2\theta_0(2Q_1(\frac{1}{2})-Q_1(1))$),
   it is optimal to choose $a_1=a_2=1$ and keep the zero path-diversity equilibrium, and the corresponding social welfare is $\theta_0Q_1(1)+\theta_0Q_0$.
\end{itemize}
By searching over all possible parameters and information value function $Q_1(\cdot)$, the mechanism achieves $PoA_a=1/3$.
\end{proposition}
\begin{proof}
 When we apply content-restriction as incentive, at the equilibrium,  there are three possibilities: 1) Only one path is used by the users. 2) Exactly two paths are used by the users. 3) All three paths are used by the users. We analyse these three case in the following:

\emph{Analysis of 1)}: Since only one path is used, the content collected hence is $Q_1(1)+Q_0$. Note that we can simply adopt the original equilibrium and avoid any loss of social welfare due to destruction of content. Thus, we optimally choose $a_1=a_2=1$ and the corresponding social welfare at the equilibrium is $\theta_0Q_1(1)+\theta_0Q_0$.

\emph{Analysis of 2)}: If we expect exactly two paths are used at the equilibrium, we need to make the payoffs from choosing these two paths be equal and strictly larger than the payoff from choosing the third path. Among all possible pairs of paths, we always prefer path $P_1$ and path $P_2$. This is because to achieve the same path diversity, less content needs to be restricted when we only use $a_1<1$ and $a_2=1$ to make the payoffs from choosing paths $P_1$ and $P_2$ the same. Now this case is equivalent to using content-restriction as incentive in two-path model. According to the results from Section \ref{contentrestriction}, we optimally choose $a_1=1-c_2/(2\theta_0Q_1'(1/2)+\theta_0Q_0-\epsilon)$ and $a_2=1$ to approach the perfect path diversity between path $P_1$ and path $P_2$ $(0,1/2,1/2)$ and optimal social welfare is $2\theta_0Q_1'(1/2)+\theta_0Q_0-c_2$.

\emph{Analysis of 3)}: If we expect all the three paths are used at the target equilibrium $(\hat{x}_1,\hat{x}_2,\hat{x}_3)$, we need to choose proper $a_1$ and $a_2$ to make the payoffs from choosing these three paths be equal:
$$a_1\theta_0Q(\hat{x}_1,\hat{x}_2,\hat{x}_3)=a_2\theta_0Q(\hat{x}_1,\hat{x}_2,\hat{x}_3)-c_2=\theta_0Q(\hat{x}_1,\hat{x}_2,\hat{x}_3)-c_3.$$
The social welfare attained at the target equilibrium $(\hat{x}_1,\hat{x}_2,\hat{x}_3)$ is
$$\theta_0Q(\hat{x}_1,\hat{x}_2,\hat{x}_3)-c_3.$$
Thus, the equilibrium we prefer is $(1/3,1/3,1/3)$, i.e., flow are partitioned equally among three paths, where the content is maximized. The content-restriction coefficients we use to incentive this perfect path-diversity equilibrium are
$$a_1=1-\frac{c_2}{3\theta_0Q_1(\frac{1}{3})+\theta_0Q_0},\quad a_2=1-\frac{c_3-c_2}{3\theta_0Q_1(\frac{1}{3})+\theta_0Q_0}.$$
The corresponding social welfare is $3\theta_0Q_1(1/3)+\theta_0Q_0-c_3$. However, $(1/3,1/3,1/3)$ is unstable. Instead, we will choose
\begin{equation}\label{alam}
a_1=1-\frac{c_3}{3\theta_0Q_1(\frac{1}{3})+\theta_0Q_0-\epsilon},\quad a_2=1-\frac{c_3-c_2}{3\theta_0Q_1(\frac{1}{3})+\theta_0Q_0-\epsilon}
\end{equation}
to incentivize a stable set of equilibria as follows:
$$E:=\{(x_1,x_2,x_3)\in S_3|Q(x_1,x_2,x_3)=3\theta_0Q_1(\frac{1}{3})+\theta_0Q_0-\epsilon,\ 2x_3+x_2>1\}$$
Because the continuity and symmetry of $Q(x_1,x_2,x_3)$, $E$ is not empty. First we define the evolution of the system. We let $\mathcal{P}_{max}(x_1,x_2,x_3)$ denote the set of most profitable paths when the current flow partition is $(x_1,x_2,x_3)$, i.e.,
$$\mathcal{P}_{max}(x_1,x_2,x_3)=\argmax\limits_{P_k\in\{P_1,P_2,P_3\}}\{u(P_k,x_1,x_2,x_3)\}.$$
 Let $P_{min}(x_1,x_2,x_3)$ denote the least profitable path when the current flow partition is $(x_1,x_2,x_3)$, i.e.,
 $$\mathcal{P}_{min}(x_1,x_2,x_3)=\argmin\limits_{P_k\in\{P_1,P_2,P_3\}}\{u(P_k,x_1,x_2,x_3)\}.$$
 We denote a path in $\mathcal{P}_{max}$ ($\mathcal{P}_{min}$) by $P_{max}$ ($P_{min}$).  Now we define the differential equation which governs the evolution of the system by
\begin{equation}\label{de1}
\frac{\dif x_3}{\dif t}=\left\{
\begin{array}{ll}
\mu(u(P_3,x_1,x_2,x_3)-u(P_{min}(x_1,x_2,x_3),x_1,x_2,x_3)), &\mbox{if }P_3\in\mathcal{P}_{max}(x_1,x_2,x_3);\\
\mu(u(P_3,x_1,x_2,x_3)-u(P_{max}(x_1,x_2,x_3),x_1,x_2,x_3)), &\mbox{if }P_3\in\mathcal{P}_{min}(x_1,x_2,x_3);\\
0, &\mbox{else}.
\end{array}
\right.
\end{equation}
\begin{equation}\label{de2}
\frac{\dif x_2}{\dif t}=\left\{
\begin{array}{ll}
\mu(u(P_2,x_1,x_2,x_3)-u(P_{min}(x_1,x_2,x_3),x_1,x_2,x_3)), &\mbox{if }P_2\in\mathcal{P}_{max}(x_1,x_2,x_3);\\
\mu(u(P_2,x_1,x_2,x_3)-u(P_{max}(x_1,x_2,x_3),x_1,x_2,x_3)), &\mbox{if }P_2\in\mathcal{P}_{min}(x_1,x_2,x_3);\\
0, &\mbox{else}.
\end{array}
\right.
\end{equation}
where $\mu>0$ indicating the convergence rate. Thus, a user's strategy in each step is to firstly check whether his current path is the least profitable one. If her path is the least profitable  one, she will change to the most profitable path with some positive probability. Otherwise, she will keep to her current path choice.

 Note that $x_1+x_2+x_3=1$,  $Q(x_1,x_2,x_3)$ is determined given $x_2$ and $x_3$. Thus, we can write $Q(x_1,x_2,x_3)$ as a function of only $x_2$ and $x_3$, i.e., we define
 $$Q(x_2,x_3,1):=Q_0+Q_1(x_2)+Q_1(x_3)+Q_1(1-x_2-x_3).$$
 Let $I=\{(x_2,x_3)|(x_1,x_2,x_3)\in S_3\mbox{ and }2x_3+x_2>1\}$ and $E_2=\{(x_2,x_3)|(x_1,x_2,x_3)\in E\}$ it is easy to see that for all $(x_1,x_2,x_3)\in I$,
$$\frac{\partial Q(x_2,x_3,1)}{\partial x_3}<0.$$
Now we define the following Lyapunov function
$$V(x_2,x_3)=(Q(x_2,x_3,1)-3\theta_0Q_1(\frac{1}{3})-\theta_0Q_0+\epsilon)^2\quad \forall (x_1,x_2,x_3)\in I.$$
Compute the first derivative, $\forall (,x_2,x_3)\in I$,
$$\frac{\dif V(x_2,x_3)}{\dif t}=2(Q(x_2,x_3,1)-3\theta_0Q_1(\frac{1}{3})-\theta_0Q_0+\epsilon))(\frac{\partial Q(x_2,x_3,1)}{\partial x_3}\frac{\partial x_3}{\partial t}+\frac{\partial Q(x_2,x_3,1)}{\partial x_2}\frac{\partial x_2}{\partial t}).$$
Note that if $Q(x_2,x_3,1)-3\theta_0Q_1'(\frac{1}{3})-\theta_0Q_0+\epsilon>0$, we have $u(P_3,x_1,x_2,x_3)>u(P_2,x_1,x_2,x_3)>u(P_1,x_1,x_2,x_3)$. Then, according to the dynamic we define,
$$\frac{\partial x_H}{\partial t}>0,\quad\frac{\partial x_M}{\partial t}=0. $$
If $Q(x_H,x_M,1)-3\theta_0Q_1'(\frac{1}{3})-\theta_0Q_0+\epsilon<0$, we have $u(H,x_H,x_M)<u(M,x_H,x_M)<u(L,x_H,x_M)$. Then, according to the dynamic we define,
$$\frac{\partial x_3}{\partial t}<0,\quad\frac{\partial x_2}{\partial t}=0. $$
It follows that
$$\frac{\dif V(x_2,x_3)}{\dif t}<0,\quad\forall (x_2,x_3)\in I/E.$$
Since $V(x_2,x_3)=0$ for all $(x_2,x_3)\in E_2$ and $V(x_2,x_3)>0$ for all $(x_2,x_3)\in I/E_2$, $E$ is a stable set of equilibria. Note that $a_1$ and $a_2$ are fixed, value of $Q(x_1,x_2,x_3)$ are the same for all the equilibria in $E$, we conclude that we achieve the same social welfare for each equilibrium in $E$. As $\epsilon\rightarrow0$, the social welfare obtained at any equilibrium in $E$ goes to $\theta_0Q(\hat{x}_1,\hat{x}_2,\hat{x}_3)-c_3$.

Now we only need to compare the optimal social welfare attained in the three cases:
$$\theta_0Q_1(1)+\theta_0Q_0,\quad 2\theta_0Q_1(\frac{1}{2})+\theta_0Q_0-c_2,\quad 3\theta_0Q_1(\frac{1}{3})+\theta_0Q_0-c_3,$$
and choose proper $a_1$ and $a_2$ to incentivize the equilibrium with the largest social welfare.
Note that $(c_1,c_2,c_3)=(0,c_H/2,c_H)$. If $c_H> 2\theta_0(2Q_1(\frac{1}{2})-Q_1(1))$,  it is optimal to choose $a_1=a_2=1$. If $2\theta_0(3Q_1(\frac{1}{3})-2Q_1(\frac{1}{2}))< c_H\leq  2\theta_0(2Q_1(\frac{1}{2})-Q_1(1)) $, we optimally choose $a_1=1-c_2/(2\theta_0Q_1(1/2)+\theta_0Q_0-\epsilon)$ and $a_2=1$. If $c_H> 2\theta_0(2Q_1(\frac{1}{2})-Q_1(1))$, we optimally choose $a_1=1-c_3/(3\theta_0Q_1(1/3)+\theta_0Q_0-\epsilon)$ and $a_2=1-(c_3-c_2)/(3\theta_0Q_1(1/3)+\theta_0Q_0-\epsilon)$.

Next we will prove that the price of anarchy is $1/3$. Note that $SW(x_1^*,x_2^*,x_3^*)\le3\theta_0Q_1(\frac{1}{3})+\theta_0Q_0$ and we can always choose $a_1=a_2=1$ and obtain the zero path-diversity equilibrium, the price of anarchy is lower-bounded as follows:
 $$PoA_a\ge\min\limits_{\theta_i,Q_1(\cdot),Q_0,c_H}\frac{\theta_0Q_1(1)+\theta_0Q_0}{3\theta_0Q_1(\frac{1}{3})+\theta_0Q_0}\ge\frac{1}{3}.$$
 Consider a specific $Q_1(\cdot)$ function as following
\begin{equation}
\label{fucdef6new}
Q_1(x)=\left\{
\begin{array}{ll}
\frac{qx}{\delta} &0\le x\le\delta;\\
q &\delta\le x\le1.
\end{array}
\right.
\end{equation}
where we let $\delta$ be near zero. Let $c_H=2\theta_0q$ and $Q_0=0$, then
$$\theta_0Q_1(1)+\theta_0Q_0=2\theta_0Q_1(\frac{1}{2})+\theta_0Q_0-c_2=3\theta_0Q_1(\frac{1}{3})+\theta_0Q_0-c_3=\theta_0q.$$
Note that $SW(x_1^*,x_2^*,x_3^*)= SW(1-2\delta,\delta,\delta)$ for infinitesimal $\delta>0$. Since $SW(1-2\delta,\delta,\delta)=3\theta_0q-\delta c_2-\delta c_3$, we have for this specific instance the price of anarchy is
$$PoA_a\le\frac{\theta_0q}{3\theta_0q-\delta c_2-\delta c_3},$$
for infinitesimal $\delta>0$. Hence $PoA_a\le1/3$. This completes the proof.
\end{proof}



\section{Extensions to Multiple user types and Dynamic information collection}\label{extensions}
\subsection{Continuous User Types}
In this subsection, we extend our mechanisms to the uniformly distributed user types in a continuous range.  Assume the users' valuation for content $\theta$ follows a uniform distribution $U[0,1]$.
\subsubsection{Routing Equilibrium without Incentive Design}
Without incentive design, a user's path choice does not depend on her type. Notice no matter which path a user of type-$\theta$ chooses, the information value she perceived is always $\theta Q(x_H,1)$ and hence her choice only depends on path cost. Therefore, the selfish routing strategy is for every user to choose the $L$-path at the equilibrium which is the same as in Proposition \ref{prop}.
\begin{proposition}
There exists a unique equilibrium for the content routing game: $\hat{x}_H=0$ and the resultant social welfare is
$$SW(0)=\int_0^1\theta Q(0,1)\dif\theta=0.5Q(0,1).$$
\end{proposition}
\subsubsection{Side-payment as Incentive}
The side-payment mechanism is defined in Section \ref{sidepayment}. Similar to (\ref{payoff2}), the payoff of a user of type-$\theta$ under the side-payment mechanism is given by
\begin{equation}
\label{payoff2a}
u_\theta(P,x_H)=\left\{
\begin{array}{ll}
\theta Q(x_H,b)-g(x_H), &\mbox{if }P=L;\\
\theta Q(x_H,b)-c_H+\frac{b-x_H}{x_H}g(x_H), &\mbox{if }P=H.
\end{array}
\right.
\end{equation}
Since side-payment does not depend on a user's type, Proposition \ref{lemma1} still holds. We only need to reconsider users' IR in the side-payment design. Given the side-payment function (\ref{function}), an equilibrium $\hat{x}_H$ satisfies IR for any user with type-$\theta$ if,
\[
\theta Q(\hat{x}_H,b)-\frac{\hat{x}_Hc_H}{b}\ge0,
\]
where $b$ is the active participation fraction ($0\le b\le1$). Since $1-b$ is the lowest type whose IR is satisfied, we have the following condition for $b\in(0,1]$:
\[
(1-b)Q(\hat{x}_H,b)-\frac{\hat{x}_Hc_H}{b}=0.
\]
We rewrite the condition such that it also holds for $b=0$ as follows:
\begin{equation}
\label{IR2}
(1-b)bQ(\hat{x}_H,b)-\hat{x}_Hc_H=0.
\end{equation}
Now we only need to properly choose $b$ and $\hat{x}_H$ to maximize the social welfare which is given by
$$\int_{1-b}^1\theta Q(\hat{x}_H,1)\dif\theta=\frac{b(2-b)}{2}Q(\hat{x}_H,b)-\hat{x}_Hc_H.$$
We denote the optimal solution to
\begin{equation}
\label{swgu}
\max\limits_{b\in[0,1]\atop x_H\in[0,b]} \frac{b(2-b)}{2}Q(x_H,b)-x_Hc_H
\end{equation}
subject to (\ref{IR2}) as $(x_{Hg},b_g)$ and the corresponding social welfare as $SW_g(x_{Hg},b_g)$.
\begin{proposition}
\label{prop11}
When the user types follows uniform distribution $U[0,1]$, the optimal side-payment design is to choose the side-payment in (\ref{function}) with $\hat{x}_H=x_{Hg}$ to achieve the optimal social welfare
$$SW_g(x_{Hg},b_g)=\frac{b_g(2-b_g)}{2}Q(x_{Hg},b_g)-x_{Hg}c_H.$$
\end{proposition}

In Section \ref{sidepayment}, the active participation fraction $b$ is either 0.5 or 1. We only need to compare the full participation case with $b=1$ with the half participation case with $b=0.5$ and choose the better one as in Theorem \ref{prop3}. When user types are continuous, we need to optimally choose the active participation fraction $b$ in the range $[0,1]$ as in Proposition \ref{prop11}.
\subsubsection{Content-restriction as Incentive}
The content-restriction mechanism is defined in Section \ref{contentrestriction}. Similar to Section \ref{contentrestriction}, if we expect positive flows with $\hat{x}_H\in(0,1)$ and $1-\hat{x}_H$ on both paths at the equilibrium, any user with type-$\theta=1-\hat{x}_H$ should be indifferent in choosing between the both paths. Her payoffs by choosing $L$-path and $H$-path are equal, i.e.,
$$(1-\hat{x}_H)Q(\hat{x}_H,1)=\frac{c_H}{(1-a)}.$$
Then all users with type $\theta$ smaller than $(1-\hat{x}_H)$ will choose $L$-path and  all users with type $\theta$ larger than $(1-\hat{x}_H)$ will choose $H$-path. Thus, the social welfare attained at equilibrium $\hat{x}_H$ is
\begin{eqnarray}
&&\int_{0}^{1-\hat{x}_H}\theta aQ(\hat{x}_H)\dif\theta+\int_{1-\hat{x}_H}^1Q(\hat{x}_H)\dif\theta-\hat{x}_Hc_H\nonumber\\
&=&0.5(Q(\hat{x}_H,1)-(1+\hat{x}_H)c_H).\nonumber
\end{eqnarray}
Now we only need to properly choose $a$ and target $\hat{x}_H$ to maximize the social welfare
\begin{equation}
\label{swau}
\max\limits_{ x_H\in(0,1)} 0.5(Q(x_H,1)-(1+x_H)c_H).
\end{equation}
We denote the optimal solution to (\ref{swau}) as $x_{Ha}$ and the corresponding social welfare is $SW_a(x_{Ha})$. Then, the content-restriction is to compare $SW_a(x_{Ha})$ and the equilibrium social welfare with $a=1$, i.e., $0.5Q(0,1)$ to decide to incentivize $\hat{x}_H=x_{Ha}$ or $\hat{x}_H=0$.
\begin{proposition}
\label{prop12}
When the user types follows uniform distribution $U[0,1]$, depending on the relationship between $SW_a(x_{Ha})$ and $0.5Q(0,1)$, we decide $a$ and which equilibrium $\hat{x}_H$ to incentivize:
\begin{itemize}
  \item If the travel cost over $H$-path is small, i.e.,
  $$c_H<\frac{Q(x_{Ha},1)-Q(0,1)}{1+x_{Ha}},$$
  it is optimal to choose
  $$a=1-\frac{c_H}{(1-x_{Ha})Q(x_{Ha},1)}$$
  to reach the optimal social welfare
  $$SW_a(x_{Ha})=0.5(Q(x_{Ha},1)-(1+x_{Ha})c_H).$$
  \item If the travel cost is large, i.e.,
  $$c_H\ge\frac{Q(x_{Ha},1)-Q(0,1)}{1+x_{Ha}},$$
  it is optimal to choose $a=1$ since content destruction must be too excessive to incentivize path diversity. The corresponding social welfare is $0.5Q(0,1)$.
\end{itemize}
\end{proposition}
Similar to Section \ref{contentrestriction}, when the travel cost $c_H$ is too high, we may skip the path-diversity, as it is too expensive for any user type to cover the $H$-path. When the travel cost $c_H$ is small, different from two user types case, the optimal equilibrium is not always the perfect path-diversity equilibrium $\hat{x}_H=0.5$ and it is determined by the specific value of $c_H$.
\subsection{Dynamic Information Model}\label{repeated}
In this subsection we no longer consider the one-shot content routing game as in Section \ref{systemmodel}, we study the dynamic version of our content routing game. Here users need to make a travel decision in each time slot and such decision-making is myopic and happens repeatedly. In each time slot $t$, each user chooses a path $H$ or $L$. We denote by $x_{Ht}$ the fraction of the users that choose the $H$-path at the $t$th ($t\in\mathbb{N}$) time period and leaving $1-x_{Ht}$ to use the $L$-path.

We assume each path has $N/2$  independent information pieces with no information overlap between paths. At time 0, there are $Q_0=Q_{H0}+Q_{L0}$ independent pieces of information which is known to every user where $Q_{H0}$ is the amount of information on $H$-path and $Q_{L0}$ is the amount of information on $L$-path. We denote by $\Delta Q_{Ht}$ and the expected number of independent pieces of information collected over $H$-path during the $t$th time period. According to (\ref{Q}) we have
\begin{equation}
\Delta Q_{Ht}=\frac{N}{2}\big(1-(1-\frac{2\phi}{N})^{nx_{Ht}}\big).
\end{equation}
Similarly, we denote by $\Delta Q_{Lt}$ and the expected number of independent pieces of information collected over $L$- respectively during the $t$th time period, then
\begin{equation}
\Delta Q_{Lt}=\frac{N}{2}\big(1-(1-\frac{2\phi}{N})^{n(1-x_{Ht})}\big).
\end{equation}
We denote by $Q_t$ the quantity of information available to the users at time $t$. It consists of the initial information $Q_0$ and the information collected during $t$ time periods. The information collected by the users may not be useful all the time, we need to eliminate the outdated information. We model this fact by a information discount factor  $\gamma\in(0,1)$. We suppose the quantity of the useful information decreases at a fixed rate of $\gamma$. For example, if at time $t-1$, there are $Q_{t-1}$ amount of useful information in the information pool, then at time $t$ only $\gamma Q_{t-1}$ amount of information are still useful. We denote by $Q_{Ht}$ the quantity of useful information on $H$-path at time $t$, then
\begin{equation}
\label{qht}
Q_{Ht}=\frac{N}{2}(1-(1-\frac{2\gamma Q_{H(t-1)}}{N})(1-\frac{2\Delta Q_{Ht}}{N})).
\end{equation}
 We denote by $Q_{Lt}$ the quantity of useful information on $L$-path at time $t$, then
\begin{equation}
\label{qlt}
Q_{Lt}=\frac{N}{2}(1-(1-\frac{2\gamma Q_{L(t-1)}}{N})(1-\frac{2\Delta Q_{Lt}}{N})),
\end{equation}
and $Q_t=Q_{Ht}+Q_{Lt}$.
We assume the users are non-atomic and have same valuations for content $\theta=0.5$. We assume the users are myopic, i.e., they only care about their payoff at the current period. If a user with type $\theta$ chooses path $P\in\{H,L\}$ at time slot $t$, her payoff is the difference between her perceived information value $\theta Q_t$ and the travel cost on path $P$. That is,
\begin{equation}
\label{payoff4}
u_\theta(P)=\left\{
\begin{array}{ll}
\theta Q_t, &\mbox{if }P=L;\\
\theta Q_t-c_H, &\mbox{if }P=H.
\end{array}
\right.
\end{equation}
and the user will prefer the path with a higher payoff.
In the following, we will characterize the equilibrium of the dynamic routing game and the stationary state.
\subsubsection{Routing Equilibrium without Incentive Design and Social Optimum}
Notice that no matter which path a user chooses at time $t$, the information value she perceived is always $\theta Q_t$ and hence her choice only depends on the path cost. Therefore, the selfish routing strategy is for every user to choose the $L$-path at the equilibrium at any time $t$.
\begin{proposition}
At any time $t$, there exists a unique equilibrium for the content routing game: $\hat{x}_{Ht}=0$. At the stationary state, the unique equilibrium for the content routing game is $\hat{x}_{H\infty}=0$ and the corresponding social welfare is
$$\widehat{SW}_\infty=\frac{N(1-r)}{4(1-\gamma r)},$$
where
$$r=(1-\frac{2\phi}{N})^n.$$
\end{proposition}
\begin{proof}
We only need to derive the stationary social welfare for the equilibrium. During the $t$th time period, the information collected at the equilibrium is
$$\Delta \hat{Q}_{Ht}=0,\ \Delta \hat{Q}_{Lt}=\frac{N(1-r)}{2}.$$
By (\ref{qht}) and (\ref{qlt}), we have the following recursions to calculate the total information available at time $t$:
$$\hat{Q}_{Ht}=\gamma \hat{Q}_{H(t-1)},$$
$$\hat{Q}_{Lt}=\gamma \hat{Q}_{L(t-1)}+\frac{N(1-r)}{2}.$$
Solve the two recursions, we have taht at the stationary state the information available on each path is
$$\hat{Q}_{H\infty}=0,\ \hat{Q}_{L\infty}=\frac{N(1-r)}{2(1-\gamma r)}.$$
Thus we can calculate the stationary social welfare for the equilibrium:
$$\widehat{SW}_\infty=\frac{N(1-r)}{4(1-\gamma r)}.$$
\end{proof}
If we let $\gamma\rightarrow1$, we have $\hat{Q}_{H\infty}=0$ and $\hat{Q}_{L\infty}=N/2$, i.e., we can have the total information on $L$-path but no information on $H$-path if the information do not change with time.

If we can perfectly control all users' decisions in a centralized way, we can achieve the social optimum which is shown in the next proposition.
\begin{proposition}
The optimal stationary flow to achieve the social optimum is $x_{H\infty}^*$. $x_{H\infty}^*$ is equal to $0$ if
$$\frac{(r-1)N\ln r}{(1-\gamma r)}\le4c_H,$$
Otherwise, $x_{H\infty}^*$ is the solution to the following equation concerning $x_H$:
$$(1-\gamma)(\frac{r^{1-x_H}}{1-\gamma r^{1-x_H}}-\frac{r^{x_H}}{1-\gamma r^{x_H}})N\ln q=4c_H.$$
At the stationary state, the information available on $H$-path is
$$Q_{H\infty}^*=\frac{N}{2}\frac{1-r^{x_{H\infty}^*}}{1-\gamma r^{x_{H\infty}^*}},$$
the information available on $L$-path is
$$Q_{L\infty}^*=\frac{N}{2}\frac{1-r^{1-x_{H\infty}^*}}{1-\gamma r^{1-x_{H\infty}^*}}.$$
The optimal stationary social welfare is
\begin{equation}
\small
SW_\infty^*=\frac{N}{4}(\frac{1-r^{x_{H\infty}^*}}{1-\gamma r^{x_{H\infty}^*}}+\frac{1-r^{1-x_{H\infty}^*}}{1-\gamma r^{1-x_{H\infty}^*}})-x_{H\infty}^*c_H.
\nonumber
\end{equation}
\end{proposition}
\begin{proof}
Since users are myopic, we optimize the social welfare for each time period separately. The social welfare optimization problem is formulated as
$$\max\limits_{x_{Ht}\in[0,1]}\frac{1}{2}Q_{t}(x_{Ht},Q_{H(t-1)},Q_{L(t-1)})-x_{Ht}c_H.$$
The objective function is a concave function of $x_{Ht}$. First order condition gives
$$\ln r(\frac{N}{2}-\gamma Q_{H(t-1)})r^{1-x_{Ht}}-\ln r(\frac{N}{2}-\gamma Q_{H(t-1)})r^{x_{Ht}}=2c_H.$$
If $\ln r(\frac{N}{2}-\gamma Q_{H(t-1)})r-\ln r(\frac{N}{2}-\gamma Q_{H(t-1)})\le 2c_H$, optimal solution is $x_{Ht}=0$. Otherwise, the optimal solution is the unique solution of the above equation. Let $x_{Ht}^*$ denote the optimal solution. The optimal social welfare at time $t$ is
$$SW_{Ht}^*=\frac{1}{2}Q_{t}(x_{Ht}^*,Q_{H(t-1)},Q_{L(t-1)})-x_{Ht}^*c_H.$$
Next we will find the stationary state for social optimum. Since we have the following recursions:
$$Q_{Ht}^*=\frac{N}{2}-(\frac{N}{2}-\gamma \hat{Q}_{H(t-1)})r^{x_{Ht}^*},$$
$$Q_{Lt}^*=\frac{N}{2}-(\frac{N}{2}-\gamma \hat{Q}_{L(t-1)})r^{1-x_{Ht}^*},$$
it follows that
$$Q_{H\infty}^*\frac{N}{2}\frac{1-r^{x_H}}{1-\gamma r^{x_{H\infty}^*}},$$
$$Q_{L\infty}^*=\frac{N}{2}\frac{1-r^{x_L}}{1-\gamma r^{1-x_{H\infty}^*}},$$
where $x_{H\infty}^*$ is the optimal stationary flow. By first order condition,
$$\ln r(\frac{N}{2}-\gamma Q_{L\infty}^*)r^{1-x_{H\infty}^*}-\ln r(\frac{N}{2}-\gamma Q_{H\infty}^*)r^{x_{H\infty}^*}=2c_H.$$
Plug in $Q_{H\infty}^*$ and $Q_{L\infty}^*$, we get
$$(1-\gamma)(\frac{r^{1-x_{H\infty}^*}}{1-\gamma r^{1-x_{H\infty}^*}}-\frac{r^{x_{H\infty}^*}}{1-\gamma r^{x_{H\infty}^*}})N\ln r=4c_H.$$
Thus when
$$\frac{(r-1)N\ln r}{(1-\gamma r)}\le4c_H,$$
the optimal stationary flow is $x_{H\infty}^*=0$; otherwise, $x_{H\infty}^*$ is the unique solution to the following equation concerning $x_H$:
$$(1-\gamma)(\frac{r^{1-x_H}}{1-\gamma r^{1-x_H}}-\frac{r^{x_H}}{1-\gamma r^{x_H}})N\ln q=4c_H.$$
\end{proof}
If we let $\gamma\rightarrow1$, we have $Q_{H\infty}^*=N/2$ and $Q_{L\infty}^*=N/2$, i.e., we have the total information on both paths by controlling all users' decisions in a centralized way if the information do not change with time.
\subsubsection{Content-restriction as Incentive}
We consider using content-restriction as incentive. Assume at any time $t$, the system planner provides a fraction $a(t)$ of the total information, i.e., $a(t)Q_t$, to $L$-path participants. $a(t)$ may changes with time $t$. The intuition is the same as in the one-shot game (Section \ref{contentrestriction}). As long as the $H$-path cost $c_H$ is not too large, we should optimally choose an apropriate value of $a$ to incentivise the perfect path-diversity equilibrium ($\hat{x}_{H\infty}=0.5$). If the travel cost $c_H$ is too high, we may skip the path-diversity, as it is too expensive for any user type to cover the $H$-path. The following proposition shows how to choose a stationary $a$ to achieve the optimal stationary social welfare.
\begin{proposition}
If the travel cost over $H$-path is small, i.e.,
  $$c_H<\frac{N(1+\gamma r^{\frac{1}{2}})(1-r^{\frac{1}{2}})^2}{4(1-\gamma r^{\frac{1}{2}})(1-\gamma r)},$$
  we optimally choose
  $$a=1-\frac{2c_H(1-\gamma r^\frac{1}{2})}{N(1-r^\frac{1}{2})},$$
  to reach the perfect path-diversity equilibrium $\hat{x}_H=0.5$ and the optimal stationary social welfare
  $$SW_{a\infty}=\frac{N(1-r^{\frac{1}{2}})}{2(1-\gamma r^{\frac{1}{2}})}-c_H.$$
  Otherwise it is optimal to choose $a=1$ since content destruction must be too excessive to incentivize path diversity. The corresponding stationary social welfare is
  $$SW_{a\infty}=\frac{N(1-r)}{4(1-\gamma r)}.$$
\end{proposition}
\begin{proof}
We only need to find the optimal positive path-diversity equilibrium and compare it with the routing equilibrium without incentive design. At any time $t$, if we expect positive flows with $\hat{x}_{Ht}\in(0,1)$ and $1-\hat{x}_{Ht}$ on the two paths at the equilibrium, a user's payoff by choosing $L$- and $H$-path are equal, i.e.,
$$\frac{1}{2}Q_t-c_H=\frac{\alpha(t)}{2}Q_t.$$
Since each user's payoff by choosing either path is the same which is
$$\frac{1}{2}Q_t-c_H.$$
Then the social welfare can be written as
$$SW_{at}=\frac{1}{2}Q_t-c_H.$$
We want to maximize such social welfare, first order condition gives
$$\frac{1}{2}\ln r(\frac{N}{2}-\gamma Q_{H(t-1)})r^{x_H}=\frac{1}{2}\ln r(\frac{N}{2}-\gamma Q_{L(t-1)})r^{1-x_H}.$$
Let $x_{Ht}^a$ denote the optimal solution, then
$$x_{Ht}^a=\frac{1}{2}+\frac{\ln\frac{N-2\gamma Q_{L(t-1)}}{N-2\gamma Q_{H(t-1)}}}{2\ln r}.$$
The content-restriction coefficient to incentivize this optimal flow is
$$a(t)=1-\frac{2c_H}{Q_t(x_{Ht}^a,Q_{H(t-1)},Q_{L(t-1)})}\ .$$

Next we will find the stationary state. Let $Q_{H\infty}^a,Q_{L\infty}^a$ denote the stationary quantity of information of on $H$-path and $L$-path respectively. Let $x_{H\infty}^a$ denote the stationary flow on $H$-path. We have the following equations
$$Q_{H\infty}^a=\frac{N}{2}\frac{1-r^{x_{H\infty}^a}}{1-\gamma r^{x_{H\infty}^a}},\ Q_{L\infty}^a=\frac{N}{2}\frac{1-r^{1-x_{H\infty}^a}}{1-\gamma r^{1-x_{H\infty}^a}},$$
$$\ln r(\frac{N}{2}-\gamma Q_{L\infty}^a)r^{1-x_{H\infty}^a}-\ln r(\frac{N}{2}-\gamma Q_{H\infty}^a)r^{x_{H\infty}^a}=0.$$
Solve the equations we get
$$x_{H\infty}^a=\frac{1}{2},\ Q_{H\infty}^a=\frac{N}{2}\frac{1-r^{\frac{1}{2}}}{1-\gamma r^{\frac{1}{2}}},\ Q_{L\infty}^a=\frac{N}{2}\frac{1-r^{\frac{1}{2}}}{1-\gamma r^{\frac{1}{2}}}.$$
Then the optimal stationary social welfare attained at this perfect path-diversity equilibrium is
$$\frac{N(1-r^{\frac{1}{2}})}{2(1-\gamma r^{\frac{1}{2}})}-c_H,$$
the $a$ we should choose to achieve such perfect path-diversity equilibrium is
$$a=1-\frac{2c_H(1-\gamma r^\frac{1}{2})}{N(1-r^\frac{1}{2})}.$$
\end{proof}

\section{Conclusion}\label{conclusion}
This paper studies the incentives of participation and route selection in information sharing system by combining information sharing with routing. We consider a content routing game where a unit mass of non-atomic selfish users choose between a high-cost and a low-cost path. To remedy the inefficient single path equilibrium of the content routing game, we design two
 incentive mechanisms to induce path diversity: side-payment and content-restriction. 
We show that both mechanisms achieve path diversity in routing choices at the cost of user participation or content destruction. We also show that user diversity can have opposite effects on the two mechanisms.
 We combine the above-mentioned two mechanisms and show that the resulting mechanism is much better than any of them used alone with $PoA_{ag}\ge0.7$. Finally, we generalize the results in the previous sections by considering a more general network model and show that similar techniques can be used to derive the optimal incentive schemes.

There are some possible ways to extend this work. For example, we can add a traffic-dependent congestion cost to users' payoffs depending on their path choices and the traffic there, and investigate how the negative externalities due to traffic congestion will interplay with the positive externalities of information sharing to decide users' routing.  In Appendix \ref{app:A} we have shown that our mechanism design can be easily applied to a linear traffic-dependent cost model and all the main results including price of anarchy still hold.


 Another possible generalization is that some users may value the information along one path differently than on the other path (say because they actually live near that path and may use this information more).
Hence now a user type is defined by the value $(\theta_H,\theta_L)$ of the content valuation parameters along each of the paths. In this case the value of content becomes
$\theta_HQ_1(x)+\theta_LQ_1(1-x)$ and it generalises the case we studied earlier where $\theta_L=\theta_H$. Our analysis and mechanism design can also be easily applied to this non-symmetric case as shown in Appendix \ref{app:B}.

 Finally, it is also interesting to extend the one-shot routing game to a repeated game, where users repeat sensing trips over time and prior content may be useful later on. Some preliminary results are also shown in Section \ref{repeated}.
\appendix
\subsection{Extension to Traffic-dependent Cost}\label{app:A}
We assume in the previous sections that the travel cost is constant for each path to avoid tedious analysis. However, one may argue that in reality the traffic flow also affects the travel cost. To address this concern, we will consider a linear cost model in this section. We will show that the side-payment mechanism and content-restriction can still be applicable to the new cost model.

Given the traffic flow partition $(x_H, 1-x_H)$ between $H$-path and $L$-path, travelling on the $H$-path incurs a travel cost given by
$$c_H+b_Hx_H,$$
where $c_H, b_H\ge0$. Similarly, travelling on the $L$-path incurs a travel cost given by
$$c_L+b_L(1-x_H),$$
where $c_L, b_L\ge0$. Note that when $c_H>0$ and $b_H=c_L=b_L=0$, this is exactly the constant cost model we have considered in the previous sections.

We assume all users  are of the same type (i.e., $\theta_1=\theta_2=\theta_0$). 
Similar to (5) in Section II, here the payoff  function of a user is
\begin{equation}
\label{linearcostpayoff}
u(P,x_H)=\left\{
\begin{array}{ll}
\theta_0Q(x_H,1)-(c_L+b_L(1-x_H)), &\mbox{if }P=L;\\
\theta_0Q(x_H,1)-(c_H+b_Hx_H), &\mbox{if }P=H.\nonumber
\end{array}
\right.
\end{equation}
depending on her path choice $P\in\{L, H\}$.
The social welfare now becomes
$$SW(x_H)=\theta Q(x_H,1)-(b_H+b_L)x_H^2+(2b_L+c_L-c_H)x_H-(c_L+b_L).$$

Following the logic of previous sections, we will compute the Nash equilibrium and social optimum firstly. If they are not the same, we will design an appropriate side-payment function to remedy the Nash equilibrium such that under the side-payment function the Nash equilibrium coincides with the socila optimum. Because the reasoning process is lengthy and very similar to the previous parts, we will omit it here. Basically, the methodology we use is very standard and straightforward.

\subsubsection{Side-payment as Incentive}
The idea of using side-payment as incentive in this linear cost model is exactly the same as in the constant cost model: charge more users
that take the $L$-path and providing positive subsidies to users of the $H$-path. Only some technical details may be different. Before state our main results, we introduce two equations. Define $\Delta c:=c_H-c_L$. The first one is
\begin{equation}
\label{derivativeofSW}
\theta_0Q_1^\prime(x_H)-\theta_0Q_1^\prime(1-x_H)-2(b_H+b_L)x_H+(2b_L+c_L-c_H)=0.
\end{equation}
This equation gives the value of $x_H$ such that the derivative of $SW(x_H)$ is zero. Note that when $\Delta c\in(-b_H,b_L)$, Equation \eqref{derivativeofSW} has a unique solution which also denotes the social optimal flow.
The second equation is
\begin{equation}
\label{valueofdeltaa}
\theta_0Q_1^\prime(\frac{b_L-\Delta c}{b_H+b_L})-\theta_0Q_1^\prime(\frac{b_H+\delta c}{b_H+b_L})+\Delta c=0
\end{equation}
This equation gives a value of $\Delta c$ such that the social optimum coincides with the Nash equilibrium. Note that Equation \eqref{valueofdeltaa} has a unique solution in $(-b_H,b_L)$ and we denote it by $\Delta \tilde{a}$. Also note that $\Delta \tilde{a}>0$ when $b_L>b_H$, $\Delta \tilde{a}=0$ when $b_L=b_H$ and $\Delta \tilde{a}=0$ when $b_L<b_H$. As $\Delta \tilde{a}$ can be any real number, we divide its range into several sets listed in Table \ref{rangeofdeltaa}.


\begin{table}[!t]
\scriptsize
\centering
\caption{Range of $\Delta c$}
\begin{tabular}[tb]{|c|c|}\hline
Notation& Associated range of $\Delta c$ \\\hline
$A_1$&$(-\infty,\theta_0Q_1^\prime(1)-\theta_0Q_1^\prime(0)-2b_H]$\\\hline
$A_2$&$(\theta_0Q_1^\prime(1)-\theta_0Q_1^\prime(0)-2b_H,-b_H]$\\\hline
$A_3$&$(-b_H,\Delta \tilde{a})$\\\hline
$A_4$&$\{\Delta \tilde{a}\}$\\\hline
$A_5$&$(\Delta \tilde{a},b_L)$\\\hline
$A_6$&$[b_L,\theta_0Q_1^\prime(0)-\theta_0Q_1^\prime(1)+2b_L)$\\\hline
$A_7$&$[\theta_0Q_1^\prime(0)-\theta_0Q_1^\prime(1)+2b_L,\infty)$\\\hline
\end{tabular}
\label{rangeofdeltaa}
\end{table}

\begin{table}[!t]
\scriptsize
\centering
\caption{Side-payment function for the linear cost model}
\begin{tabular}[tb]{|c|c|c|c|c|c|}\hline
Range of $\Delta c$&Nash equilibrium $\hat{x}_H$&Social optimum $x_H^*$&$\hat{x}_H?x_H^*$&Side-payment function\\\hline
$A_1$&1&1&$\hat{x}_H=x_H^*$&$g_{x_H^*}(x_H)=0$\\\hline
$A_2$&1&solution to \eqref{derivativeofSW}&$\hat{x}_H>x_H^*$&$g_{x_H^*}(x_H)=\bar{g}(x_H^*)+\frac{\bar{g}(x_H^*)}{x_H^*}(x_H-x_H^*)$\\\hline
$A_3$&$\frac{c_L+b_L-c_H}{b_H+b_L}$&solution to \eqref{derivativeofSW}&$\hat{x}_H>x_H^*$&$g_{x_H^*}(x_H)=\bar{g}(x_H^*)+\frac{\bar{g}(x_H^*)}{x_H^*}(x_H-x_H^*)$\\\hline
$A_4$&$\frac{c_L+b_L-c_H}{b_H+b_L}$&solution to \eqref{derivativeofSW}&$\hat{x}_H=x_H^*$&$g_{x_H^*}(x_H)=0$\\\hline
$A_5$&$\frac{c_L+b_L-c_H}{b_H+b_L}$&solution to \eqref{derivativeofSW}&$\hat{x}_H<x_H^*$&$g_{x_H^*}(x_H)=\bar{g}(x_H^*)+\frac{\bar{g}(x_H^*)}{x_H^*-1}(x_H-x_H^*)$\\\hline
$A_6$&0&solution to \eqref{derivativeofSW}&$\hat{x}_H<x_H^*$&$g_{x_H^*}(x_H)=\bar{g}(x_H^*)+\frac{\bar{g}(x_H^*)}{x_H^*-1}(x_H-x_H^*)$\\\hline
$A_7$&0&0&$\hat{x}_H=x_H^*$&$g_{x_H^*}(x_H)=0$\\\hline
\end{tabular}
\label{tableofresults}
\end{table}

We summarise all the results in Table \ref{tableofresults}. In Table \ref{tableofresults}, the function $\bar{g}(\cdot)$ is given by
 $$\bar{g}(x_H^*)=(c_H-c_L)x_H^*-b_Lx_H^*+(b_H+b_L){x_H^*}^2.$$
 It is the side payment each user of $L$-path actually pay to ensure that social optimum coincides with Nash equilibrium and hence it is a function of the social optimal flow $x_H^*$. Note that $\bar{g}(x_H^*)>0$ when $\hat{x}_H<x_H^*$ and $\bar{g}(x_H^*)<0$ when $\hat{x}_H>x_H^*$. This is in accordance with our intuition that when the flow on $H$-path is lower than the optimal flow we should penalise the users who choose $L$-path and when the flow on $H$-path is higher than the optimal flow we should reward the users who choose $L$-path.

 We can also prove that the price of anarchy without incentive design is $1/2$ and the side payment mechanism for homogenous user types achieves $PoA_g=1$. This claim follows trivially  from an argument analogous to the one we used in the proof of Proposition 2.
 \subsubsection{Content-restriction as Incentive}
 The main ideas of using content-restriction as incentive are exactly the same as the analogous part in previous sections: to motivate more users to choose the $H$-path, the system planner provides a fraction $a_L$ of the total information, i.e,  $a_LQ(x_H,1)$,
 to $L$-path participants; similarly, to motivate more users to choose the $L$-path, the system planner provides a fraction $a_H$ of the total information, i.e,  $a_HQ(x_H,1)$,
 to $H$-path participants. We will state our results directly since and the technical details are very tedious.  Before state our main results, we introduce two notations. We define $x_{b_L}$ as
 $$x_{b_L}\in\argmax_{0\le x\le 1}\theta_0Q(x,1)+b_Lx-c_L-b_L.$$
 Note that $x_{b_L}$ gives the value of the flow on $H$-path such that the payoff of choosing $L$-path is maximized. Also note that if $\theta_0Q^\prime(1)\ge -b_L$ then $x_{b_L}=1$ and if $\theta_0Q^\prime(1)<-b_L$ then $x_{b_L}\in(1/2,1)$. Analogously, we define $x_{b_H}$ as
 $$x_{b_H}\in\argmax_{0\le x\le 1}\theta_0Q(x,1)-b_Hx-c_H.$$
 Note that $x_{b_H}$ gives the value of the flow on $H$-path such that the payoff of choosing $H$-path is maximized. Also note that if $\theta_0Q^\prime(1)\le b_H$ then $x_{b_H}=0$ and if $\theta_0Q^\prime(1)>b_H$ then $x_{b_L}\in(0,1/2)$.
 Based on which domain $\Delta c$ is in, our analysis is divided into the following cases:

\emph{Case 1: $\Delta c\in A_1\cup A_4\cup A_7$.} As $\hat{x}_H=x_H^*$ in this case, we don't need to
motivate more users to choose the $L$-path or $H$-path. Hence, we let $a_H=a_L=1$.

 \emph{Case 2: $\Delta c\in A_2\cup A_3$.} As $\hat{x}_H>x_H^*$ in this case, we need to
motivate more users to choose the $L$-path. Hence, we let $a_L=1$ and $a_H\le1$.

Furthermore, if $\Delta c<b_L-(b_H+b_L)x_{b_L}$ and $\theta_0Q(x_{b_L})+b_Lx_{b_L}-b_L-c_L>SW(\hat{x}_H)$, it is optimal to choose
$$a_H=1-\frac{ b_L-(b_H+b_L)x_{b_L}-\Delta c}{\theta_0Q(x_{b_L},1)},$$
and the corresponding social welfare is
$$\theta_0Q(x_{b_L})+b_Lx_{b_L}-b_L-c_L.$$
Otherwise, it is optimal to choose $a_H=1$ and the corresponding social welfare is same as the Nash equilibrium.

 \emph{Case 3: $\Delta c\in A_5\cup A_6$.} As $\hat{x}_H<x_H^*$ in this case, we need to
motivate more users to choose the $H$-path. Hence, we let $a_H=1$ and $a_L\le1$.

Furthermore, if $\Delta c>b_L+(b_H+b_L)x_{b_H}$ and $\theta_0Q(x_{b_H})-b_Hx_{b_H}-c_H>SW(\hat{x}_H)$, it is optimal to choose
$$a_L=1-\frac{\Delta c+(b_H+b_L)x_{b_H}-b_L}{\theta_0Q(x_{b_H},1)},$$
and the corresponding social welfare is
$$\theta_0Q(x_{b_H},1)-b_Hx_{b_H}-c_H.$$
Otherwise, it is optimal to choose $a_L=1$ and the corresponding social welfare is same as the Nash equilibrium.

We can prove that the content-restriction mechanism for homogenous user types achieves $PoA_a=1/2$. This claim follows trivially  from an argument analogous to the one we used in the proof of Theorem 2.

By now, we only consider the case that user types are homogenous. Actually, our mechanism design can be applied  to heterogeneous users types as well and the main idea is analogous to the constant cost model. To avoid tedious discussion, we omit it here.

\subsection{Extension to Nonsymmetric Content Function}\label{app:B}
We assume in the previous sections that the content on both paths are of equally importance to any user with the same type. However, one may argue that one may value the content one path differently from the other path. To address this concern, we will consider a nonsymmetric content function in this section.

Given the traffic flow partition $(x_H, 1-x_H)$ between $H$-path and $L$-path, the content vaue function of a user of type $\beta$ is
$$Q_1(x_H)+\beta Q_1(1-x_H),$$
where $\beta\in(0,1)$. This means this user of type $\beta$ value the content on $H$-path more than the content on $L$-path.

We assume all users  are of the same type $\beta$. 
Similar to (5) in Section II, here the payoff  function of a user is
\begin{equation}
\label{linearcostpayoff}
u(P,x_H)=\left\{
\begin{array}{ll}
Q_1(x_H)+\beta Q_1(1-x_H), &\mbox{if }P=L;\\
Q_1(x_H)+\beta Q_1(1-x_H)-c_H, &\mbox{if }P=H.\nonumber
\end{array}
\right.
\end{equation}
depending on her path choice $P\in\{L, H\}$.
The social welfare now becomes
$$SW(x_H)=Q_1(x_H)+\beta Q_1(1-x_H)-c_Hx_H.$$

Actually, the analysis and results of the previous sections of our paper can be applied to this model without much efforts. Particularly, the price of anarchy results are exactly the same: the price of anarchy without incentive design is $1/2$,  the side payment mechanism for homogenous user types achieves $PoA_g=1$ and the content-restriction mechanism for homogenous user types achieves $PoA_a=1/2$.



\end{document}